\documentclass[preprint,12pt]{elsarticle}
\usepackage[a4paper, total={5.7in, 9in}]{geometry}

\usepackage{amssymb}
\usepackage{amsmath}
\DeclareMathOperator{\rank}{rank}
\DeclareMathOperator{\Area}{Area}
\DeclareMathOperator{\Cond}{Cond}
\DeclareMathOperator{\dom}{dom}
\DeclareMathOperator{\wdom}{wdom}
\usepackage{graphicx}
\usepackage{amsthm}
\usepackage{array}
\usepackage{hyperref}
\usepackage{algorithm}
\usepackage{algpseudocode}

\newtheorem{definition}{Definition}

\newtheorem{proposition}{Proposition}

\newtheorem{corollary}{Corollary}
\newtheorem{theorem}{Theorem}
\newtheorem{lemma}{Lemma}
\newtheorem{statement}{Statement}
\newtheorem{remark}{Remark}[section]
\newtheorem{example}{Example}

\begin{document}

\begin{frontmatter}



\title{Quantal Response Equilibrium and Rationalizability: Inside the Black Box
\footnote{We thank Pierpaolo Battigalli and Thomas Palfrey for their valuable comments and insightful suggestions, which significantly improved the paper. We thank the two anonymous reviewers, who played a crucial role in refining our work. We are grateful to Gerrit Bauch, Niels Boissennet, Yves Breitmoser, Samuele Dotta, Manuel F\"orster, Pierfrancesco Guarino, Dominik Karos, Frank Riedel, and Jurek Preker for their helpful comments and their encouragement. We thank participants in CEPET 2023, GRASS XVII, and the seminar at Bielefeld University. Shuige Liu acknowledges the financial support of Grant-in-Aids for Young Scientists (B) of Japan Society for the Promotion of Science (No.21K132637). Fabio Maccheroni acknowledges the financial support of the Italian Ministero dell’Università e della Ricerca (grant 2017CY2NCA). \\Email: \textsf{shuige.liu@unibocconi.it} (S. Liu), \textsf{fabio.maccheroni@unibocconi.it} (F. Maccheroni)}}


\author[inst1]{Shuige Liu}

\affiliation[inst1]{organization={Bocconi University},
            addressline={Via Roentgen, 1}, 
            city={Milan},
            postcode={20136}, 
            state={MI},
            country={Italy}}
            
\author[inst2]{Fabio Maccheroni}
\affiliation[inst2]{organization={Bocconi University and IGIER},
            addressline={Via Roentgen, 1}, 
            city={Milan},
            postcode={20136}, 
            state={MI},
            country={Italy}}
            
\begin{abstract}
This paper aims to connect epistemic and behavioral game theory by examining  the epistemic foundations of quantal response equilibrium (QRE) in static games. We focus on how much information agents possess about the probability distributions of idiosyncratic payoff shocks, in addition to the standard assumptions of rationality and common belief in rationality. When these distributions are transparent, we obtain a solution concept called $\Delta^p$-rationalizability, which includes action distributions derived from QRE; we also give a condition under which this relationship holds true in reverse. When agents only have common belief in the monotonicity of these distributions (for example, extreme value distributions), we obtain another solution concept called $\Delta^M$-rationalizability, which includes action distributions derived from rank-dependent choice equilibrium, a parameter-free variant of QRE. Our solution concepts also provide insights for interpreting experimental and empirical data.
\end{abstract}




\end{frontmatter}

\section{Introduction}
\label{sec:intro}
\noindent  Since the publication of the seminal work by McKelvey and Palfrey \cite{mp95}, research on quantal response equilibrium (QRE) has become  a thriving and rapidly advancing field (see Goeree, Holt, and Palfrey \cite{ghp16} for a survey). Instead of the exact best response assumed in Nash equilibrium, which causes the zero-likelihood problem in applications, by adopting a ``smooth'' (``disturbance-against-disturbance'') response, QRE has been recognized as a remarkable tool for developing empirically testable theories and fitting data. It is no surprise that the concept of QRE has been widely adopted and has prompted extensive, fruitful  experimental and theoretical research, improving our understanding of strategic behavior.

Yet one research avenue has been little explored: why people's behavior adheres to QRE.  The literature has extensively focused on interpreting and applying QRE from an observer's perspective.\footnote{There are exceptions, for example, research on cognitive hierarchy (for example, Goeree and Holt \cite{gh04}, Lin and Palfrey \cite{lp22}. See Chapter 4 in Goeree et al. \cite{ghp16} for a detailed discussion).} However, if we embrace methodological individualism as a foundational assumption in the social sciences, we must view human behavior as intentional and driven by motivation, and we should try to understand what happens inside the black box. This perspective is relevant because comprehending players' motivation and reasoning structure helps elucidate experimental data. It is especially pertinent when dealing with supposedly anomalous  data, since analyzing the functioning in the black box might help us to determine whether such data fall within the boundaries of the theoretical framework. Moreover, it aids in reassessing and refining experimental design. This approach could also provide valuable experimental insights to theorists who investigate the black box and could prompt them to explore the testability of their theories through experiments.

This paper  explores this avenue by using the tools developed in epistemic game theory (EGT).\footnote{This paper focuses on QRE in static games. The epistemic investigation of QRE in sequential games (McKelvey and Palfrey \cite{mp98}) will be conducted in another paper.} EGT analyzes strategic reasoning by studying various epistemic conditions and their behavioral consequences.\footnote{See, for example, Battigalli and Bonanno \cite{bb99}, Perea \cite{pa2012},  and Dekel and Siniscalchi \cite{ds15} for surveys. The textbook on game theory by Battigalli, Catonini, and De Vito \cite{betal20} provides comprehensive and insightful treatments permeated by EGT motivations, although it is not about formal EGT.} EGT started by questioning the  informal justifications of classical solution concepts (for example, ``each player has correct beliefs about her opponents' behavior'' for Nash equilibrium; see Aumann and Brandenburger \cite{ab95}). An important part of the EGT literature is dedicated to exploring the epistemic justification of given solution concepts, that is, the conditions regarding players' knowledge and beliefs that lead to the behavioral implications characterized by a given solution concept. This exploration relies solely on the application of decision-theoretic principles of rationality and strategic reasoning.

Since QRE is seen as capturing bounded rationality, one might naturally question its compatibility with the EGT approach, which has rationality at its core.\footnote{Though the term  \textit{rationality} has various meanings in the literature, in EGT it means maximization of an agent's expected payoff based on her subjective beliefs. See, for example, Aumann and Brandenburger \cite{ab95} and the aforementioned surveys and monographs on EGT.} Fortunately, as already elucidated in McKelvey and Palfrey \cite{mp95}, we can interpret the profile of  ``disturbance-against-disturbance'' responses designated by QRE as a mixed-action equilibrium of a static game with incomplete information.\footnote{Psychology and decision sciences have long focused on separating the phenomenon of probability-mechanism-governed behavior from  its cause and interpretation (see, for example, Luce and Suppes \cite{ls65}, chap.~5; Goeree et al. \cite{ghp16}, p.~11). For example, one may take utilities as fixed and ascribe the probability to bounded rationality influenced by unspecified psychological or cognitive factors (for example, ``trembling hands'' in Selten \cite{se75}), which is called the \emph{constant utility model}. Others explain the probability of actions by means of a randomly determined utility function, which is called the \emph{random utility model}. This model  is adopted in McKelvey and Palfrey \cite{mp95}, yielding the concept which is now called  \emph{structural QRE} (to differentiate it from the disturbance-against-disturbance one---that is, \emph{regular} QRE---in McKelvey and Palfrey \cite{mp96}). See Goeree et al. \cite{ghp16}, chap.~2.} Consider a static game being played among populations. Population $i$ has its own representative payoff function and a random variable following the distribution $p_{i}$; this variable indicates the idiosyncratic payoff shocks among individuals within the population. A QRE is composed of the action distribution in each population by pushing forward via each $p_{i}$ the best response under realized idiosyncrasies against the other populations' action distributions---that is,  the best response for one's payoff type to their belief in a game with incomplete information.\footnote{In Bayesian equilibria à la Harsanyi, what matters is the \emph{Harsanyi type}, which includes both the payoff type and a description of exogenous (hierarchical) beliefs. In this case, what matters is the payoff type because beliefs are determined by the population's action distributions.} Under this interpretation, QRE is compatible with rationality and can be analyzed with the methodologies of EGT.

Analysis of games with incomplete information from the perspective of EGT started with Battigalli and Siniscalchi \cite{bs03}.\footnote{For an introduction, see chapter~8 of Battigalli et al. \cite{betal20}. Formally, Battigalli and Siniscalchi \cite{bs03} is not an EGT paper, although it explicitly refers to epistemic game theoretical foundations presented in other works (Battigalli and Siniscalchi \cite{bs99}, \cite{bs02}, \cite{bs07}). Furthermore, it is now increasingly recognized that Battigalli \cite{b99} is the true beginning of the research on rationalizability in games with incomplete information.} In contrast to Harsanyi's \cite{ha67} canonical Bayesian-equilibrium approach, in which each agent's entire hierarchy of exogenous beliefs is (implicitly) specified, Battigalli and Siniscalchi \cite{bs03} extend Pearce's \cite{pe84} rationalizability concept and rely on explicit and general epistemic conditions of agents' knowledge and beliefs (for example, rationality, common belief in rationality), and they study outcomes (that is, pairs of payoff type and action) consistent with them. Further, their solution concept, called $\Delta$-rationalizability, accommodates also restrictions on beliefs imposed by the environment (denoted by $\Delta$).\footnote{More recently it has been relabeled as ``directed rationalizability'' to indicate the mapping from belief restrictions to the corresponding set of outcomes. See, for example, chapter~8.3 in Battigalli et al. \cite{betal20}.} By extending/generalizing Brandenburger and Dekel \cite{bd87}, they show that the set of Bayesian-equilibrium outcomes across all Harsanyi-type structures consistent with the $\Delta$ restrictions coincides with the set of $\Delta$-rationalizable outcomes. Therefore, it is convenient to adopt $\Delta$-rationalizability as the base tool and focus on the restriction on the beliefs for QRE.

Inspired by the role of the distribution profile $p = (p_{i})_{i \in I}$ of idiosyncratic payoff shocks in McKelvey and Palfrey's \cite{mp95} definition of QRE, the restriction we impose on beliefs is that $p$ is true and it is commonly believed to be true by all agents, or, in terms of EGT, $p$ is \emph{tranparent}.\footnote{The term \textit{transparency} was coined by Adam Brandenburger. We thank Pierpaolo Battigalli for this information.} The solution concept derived from rationality, common belief in rationality, and transparency of $p$ is called $\Delta^{p}$-\emph{rationalizability}, where the superscript $p$ emphasizes its cruciality. Proposition~\ref{Thmai1} shows that every QRE action distribution is $\Delta^{p}$-rationalizable, which is conceptually a special case of well-known results. Nonetheless, as is common in the literature on the epistemic foundations of equilibria, a $\Delta^{p}$-rationalizable action distribution does not necessarily appear in any QRE. Theorem~\ref{Thmai2} provides a sufficient condition for the structure of a game to guarantee that each $\Delta^{p}$-rationalizable action distribution appears in some QRE. The condition only relies on the representative payoff functions but not on $p$. When the condition is satisfied and the game has a unique QRE, the QRE can be reached by individual strategic reasoning instead of relying on a specific type space  \`a la Harsanyi \cite{ha67}. In this sense, Theorem~\ref{Thmai2} also shows when QRE is a robust prediction in the sense of Weinstein and Yildiz \cite{wy07}. \footnote{See also Bergemann and Morris \cite{bg05}. This sense of robust prediction of Bayesian Nash equilibria was already highlighted in Battigalli \cite{b99} and Battigalli and Siniscalchi \cite{bs03}.}

Several concerns may arise regarding this result. First, Theorem~\ref{Thmai2} does not put  any substantive restrictions on $p$. While this makes the result broadly applicable across various contexts, it overlooks the fact that in practice, additional assumptions are usually made about the payoff shocks. Typically, an application either assumes some functional form of regular QRE (for example, logistic quantal response functions derived from extreme value distributions) or, more often, assumes $p$ to be independent and identically distributed (i.i.d.) in the structural approach. This leads us to inquire whether the epistemic game theoretic analysis can effectively capture properties associated with such assumptions---for instance, the concept of (ordinal) monotonicity, wherein choice probabilities are ranked based on their expected payoffs. Second, the iterative procedure for obtaining $\Delta^{p}$-rationalizable outcomes (see Definition~\ref{def2}) is generically asymptotic; even in cases in which $\Delta^{p}$-rationalizability coincides with QRE, agents may have to go through a hierarchy with infinite depth. Experimental research has shown that this process can be exceptionally challenging (see, for example, Nagel \cite{nr95}, Ho, Cramer, and Weigelt \cite{hcw98}, Alaoui and Penta \cite{ap16}, Halevy, Hoelzemann, and Kneeland \cite{ hhk23}). Last, our epistemic condition requires common belief in $p$. However, except in rare instances, establishing widespread acceptance of a specific probability distribution, let alone fostering a common belief in it, is inherently difficult. \footnote{Psychological research has long recognized the cognitive challenges associated with perceiving, processing, and applying statistical information. See, for example, Cohen \cite{cj62}, Tversky and Kahneman \cite{tk71}, and Gallistel, Krishan, Liu, and Miller \cite{gklm14}.}

To address those concerns, we  relax the condition by exploring the epistemic foundation of rank-dependent choice equilibrium (RDCE) in Goeree, Holt, Louis, Palfrey, and Rogers \cite{getal19}. RDCE modifies QRE by replacing the parametric property of the latter with the coarser and simpler  monotonicity condition. In addition to $\Delta$-rationalizability, we draw inspiration from proper rationalizability (Asheim \cite{ag01}, Schumacher \cite{sf99}, Perea \cite{pa2011}) and show that the epistemic condition for RDCE is the transparency of monotonicity, or, equivalently, monotonic distributions (for example, extreme value distributions) of idiosyncratic payoff shocks in populations (alongside the standard assumptions of rationality and mutual belief in rationality). We call the solution concept $\Delta^{M}$-rationalizability, where the superscript $M$ emphasizes monotonicity. We show that each RDCE is $\Delta^{M}$-rationalizable. 

Finally, we highlight that our solution concepts target action distributions rather than actions or type-action pairs. This is another innovation of this paper. The EGT literature used to focus on rationalizable type-action pairs surviving the iterative rationalization procedure.  However, in the realm of experimental and empirical studies, prominent solution concepts like QRE and RDCE typically yield interior points, meaning that each action is rationalizable. Therefore, here, the classic concept of rationalizable (single) actions itself might not be able to provide valuable information relevant to experimental game theorists' interests. By pushing forward rationalizable action distributions from idiosyncratic payoff shocks, our results offer a convenient and relevant connection between EGT and behavioral game theory. Further, this approach might also contribute to the rapidly expanding field of empirically testing epistemic properties (See, for example, Alaoui and Penta \cite{ap16}, Kneeland \cite{tk15}, \cite{tk16}, Halevy et al. \cite{ hhk23}, Friedenberg, Kets, and Kneeland \cite{fkk21}).

The rest of the paper is organized as follows. Section~\ref{sec:qredef} gives the preliminary definitions and the epistemic conditions of QRE; also, it provides some examples and an algorithm to compute $\Delta^{p}$-rationalizable distributions. Section~\ref{sec:uncer} gives results on the condition for the coincidence of $\Delta^{p}$-rationalizability and QRE action distributions. Section~\ref{rdceep} studies the epistemic conditions of RDCE. Section~\ref{concluding} contains concluding remarks about how to use our concept to interpret data. All proofs are in Appendix A; Appendix B discusses the  generalization of Theorem~\ref{Thmai2}.

\section{Quantal response equilibrium and epistemic conditions}\label{sec:qredef}

\subsection{Quantal response equilibrium}\label{qredef}

\noindent Let $G = \langle I, (A_{j},u_{j})_{j \in I}\rangle$ be a static game, where $I$ is the finite set of players; for each $i \in I, A_{i}$ is the finite set of actions with $|A_{i}| \geq 2$, and $u_{i}:A$ $(:=\times_{j \in I}A_{j}) \rightarrow \mathbb{R}$ is the payoff function. Each $i \in I$ is interpreted as a population and $u_i$ its representative payoff function. In addition, for each $i \in I$, there is a random variable indicating the idiosyncratic payoff shock regarding each action; formally, this is modeled by associating $G$ with a pair $\langle \Theta, p\rangle$, where $\Theta := \times_{j \in I}\Theta_{j}$ and $p = (p_{j})_{j \in I}$ such that for each $i \in I$, $\Theta_{i} = \mathbb{R}^{A_{i}}$ and $p_{i}$ is a probability measure on $\Theta_{i}$.\footnote{It differs from Harsanyi's \cite{ha73} games with randomly disturbed payoffs, in which a random error is associated with each \emph{profile} of actions. Having random variables associated with individual actions, as prescribed in the setting of McKelvey and Palfrey \cite{mp95},  makes it convenient for epistemic analysis, allowing us to analyze ex ante strategic reasoning in the decision-theoretic sense.} Note that each $\theta_{i} \in \Theta_{i}$ specifies a mapping from $A_{i}$ to $\mathbb{R}$, and  for each $a_{i} \in A_{i}$, we use $\theta_{{i},a_{i}}$ to denote the value assigned to $a_{i}$ by $\theta_{i}$. Here, we only assume that each $p_{i}$ has a density function $f_{i}$, which has a marginal $f_{i,a_{i}}$ for each $a_{i} \in A_{i}$.\footnote{\label{fn1}We follow McKelvey and Palfrey \cite{mp95} in adopting these conditions on $p_{i}$ to guarantee the applicability of Brouwer fixed-point theory and the existence of QRE. The conditions imply that $p_{i}$ is absolutely continuous, which plays an important role in Theorem~\ref{Thmai2}.} Conditions assumed frequently in the literature, such as monotonicity, will be discussed in Section~\ref{rdceep}.

Given a profile of action distributions $\pi = (\pi_{j})_{j \in I} \in \times_{j\in I}\Delta(A_{j})$, for each $i \in I$ and $a_{i} \in A_{i}$, we let $\pi_{-i} = (\pi_{j})_{j \in I\setminus \{i\}}$, and we define $E_{i,a_{i}}(\pi_{-i})$ to be the set of realizations in $\Theta_{i}$ for which $a_{i}$ is a best response to $\pi_{-i}$. Formally,
\begin{equation}\label{areae}
 E_{i,a_{i}}(\pi_{-i}): = \{\theta_{i} \in \Theta_{i}:u_{i}(a_{i},\pi_{-i})+\theta_{i,a_{i}} \geq u_{i}(a_{i}^{\prime},\pi_{-i})+\theta_{i,a_{i}^{\prime}} \text{ for each } a_{i}^{\prime} \in A_{i}\},
 \end{equation}  
where $u_{i}(\cdot,\pi_{-i})$ is the expected value of $u_{i}(\cdot, a_{-i})$ with respect to $\pi_{-i}$.

Within this framework, McKelvey and Palfrey \cite{mp95} define QRE as a ``self-fulfilling'' profile of action distributions: given others' action distributions, the prescribed probability at which an action $a_{i}$ of $i$ is taken coincides with the measure with respect to $p_{i}$ of the set of the realized idiosyncrasies for which $a_{i}$ is a best response. This is formalized in Definition~\ref{def1}.

\begin{definition}\label{def1}
Given $\langle G, \Theta,p\rangle$, an action distribution profile $\pi^{*} = (\pi_{j}^{*})_{j \in I} \in \times_{j\in I}\Delta(A_{j})$ is a \textbf{quantal response equilibrium (QRE)} iff for each $i \in I$ and each $a_{i} \in A_{i}$, $\pi_{i}^{*}(a_{i})= p_{i}(E_{i,a_{i}}(\pi^{*}_{-i}))$ .
\end{definition}

One can see that QRE is the profile pushed forward from the payoff-shock distributions given a Bayesian equilibrium profile of decision functions. This push-forward method is also important for the definition of our solution concepts (Definition~\ref{norac}) and for the proof of Proposition~\ref{Thmai1}.

As mentioned in Section~\ref{sec:intro}, compared to the definition by regular quantal response function given in McKelvey and Palfrey \cite{mp96}, this definition boils down the action distributions to each agent's best response for her private information about the realized payoff idiosyncrasy. This facilitates our epistemic analysis in the following sections.

\subsection{Games with incomplete information and \texorpdfstring{$\Delta^{p}$}{delta^p}-rationalizability}\label{deltapro}
 
\noindent  The strategic environment described above can be rephrased as a \emph{game with payoff uncertainty} $\hat{G} = \langle I, (\Theta_{j},A_{j},U_{j})_{j \in I} \rangle$,  where for each $i \in I$, $\theta \in \Theta$, and $a = (a_{i},a_{-i}) \in A$, $U_{i}(\theta,a)  = u_{i}(a_{i},a_{-i}) + \theta_{i,a_{i}}$. Each $\theta_{i}$ is called a \emph{payoff type} (abbreviated as \emph{type}) of $i$. We are analyzing games with payoff uncertainty and \emph{private values}---that is, only beliefs about others' actions directly matter for expected-payoff maximization, while beliefs about others' types are relevant for strategic reasoning. This property is part of the definition of QRE and  plays a crucial role in our analysis.

Note that a game with payoff uncertainty is not a Bayesian game since it does not (implicitly or explicitly) specify the belief hierarchy for each player; rather, it could be regarded as a  general prototype that avoids imposing any implicit restriction on hierarchical beliefs.\footnote{See, for example, Battigalli and Siniscalchi \cite{bs03} and  Battigalli et al. \cite{betal20}, chap.~8.} Therefore, it is suitable for the analysis of strategic reasoning that is consistent with various epistemic conditions.

For instance, we show how to describe and obtain the behavioral implications of  two canonical epistemic assumptions, namely \emph{rationality} and \emph{common belief of rationality}.\footnote{Defining the event satisfying some epistemic conditions (for example, rationality and common belief in rationality) is different from  characterizing their behavioral implications (which is what the solution concept does). Here, our analysis is about the latter. For a formal representation of the former with  rigorous and explicit language and detailed discussions, see, for example, Battigalli and Bonanno \cite{bb99} and Dekel and Siniscalchi \cite{ds15}.} Recall that rationality means maximization of an agent's payoff based on her beliefs. At the beginning of the game, an agent in population $i$ has a (subjective) \emph{belief} $\mu^{i}$, which is a distribution over others' types and actions, that is, $\mu^{i} \in \Delta(\Theta_{-i} \times A_{-i})$. Given $\theta_{i} \in \Theta_{i}$ as the realized type, an action $a_{i} \in A_{i}$ is a \emph{best response} for $\theta_{i} \in \Theta_{i}$ to $\mu^{i}$ iff the following condition is satisfied:
\begin{equation}\label{rationality}
 \tag{\textbf{R}}
u_{i}(a_{i},\text{ marg}_{A_{-i}}\mu^{i}) + \theta_{i,a_{i}} \geq u_{i}(a_{i}^{\prime},\text{ marg}_{A_{-i}}\mu^{i}) + \theta_{i,a_{i}^{\prime}} \text{ for each }a_{i}^{\prime} \in A_{i} 
\end{equation}
Here, marg$_{A_{-i}}\mu^{i}$ is the marginal distribution of $\mu_{i}$ on $A_{-i}$, that is,  her belief about the distributions of others' actions. As we mentioned above, only beliefs about others' actions directly matter for expected payoff maximization. A pair $(\theta_{i}, a_{i})$ is \emph{consistent with rationality} iff it satisfies (\ref{rationality}) regarding some $\mu_{i}$.

Common belief of some event means that everyone believes it, and everyone believes that everyone believes it, and so on. Hence, describing it and analyzing its behavioral consequences necessitates the use of an iterative procedure. For example, ``everyone believes rationality'' means that for every agent, her belief $\mu^{i}$ only deems possible type-action pairs consistent with rationality. A pair $(\theta_{i}, a_{i})$ is \emph{consistent with rationality \emph{and} belief in rationality} iff it satisfies (\ref{rationality}) regarding some  such $\mu_{i}$. Iteratively, we can continue this procedure and define belief in belief in rationality, belief in belief in belief in rationality, and so on. A pair $(\theta_{i}, a_{i})$ is \emph{consistent with rationality and common belief in rationality} iff $a_{i}$ is a best response for $\theta_{i}$ to some belief surviving each level of the infinite hierarchy.

Rationality and common belief in rationality are the two canonical assumptions in the literature. They do not place any \emph{exogenous} restrictions on beliefs; all restrictions on beliefs are imposed by strategic thinking, that is, they are \emph{endogenous}. Battigalli and Siniscalchi \cite{bs03} examined  exogenous constraints on beliefs and investigated their implications. They called their solution concept $\Delta$-\emph{rationalizability}, which is deduced by assuming rationality, common belief in rationality, and the transparency of specified belief restrictions $\Delta$. They proved that, in static games, $\Delta$-rationalizability characterizes the set of outcomes consistent with Bayesian equilibrium across all Bayesian games that satisfy the belief restrictions $\Delta$. Since it is known that each QRE yields a Bayesian equilibrium, the problem here is to find appropriate restrictions $\Delta$ on beliefs for QRE. Through a meticulous examination of Definition~\ref{def1}, one could discern the significance of $p$, the probabilities of the idiosyncratic payoff shocks. Indeed, through $p$  the choice for each type is pushed forward into action distributions, which provides the base for everyone's best response under her type. In other words, $p$ is the engine of the fixed-point property of QRE. Therefore, it is reasonable to conjecture that $p$ provides a restriction on each individual's first-order beliefs, that is, everyone's marginal belief on each $\Theta_{j}$ should coincide with $p_{j}$.

Formally, we say that a belief $\mu^{i} \in \Delta(\Theta_{-i} \times A_{-i})$ of $i$ is \emph{consistent with $p$} iff 
\begin{equation}\label{equ1}
\tag{\textbf{CP}}
\mu^{i} = \prod_{j \neq i} \mu^{i}_{j} \text{ with } \mu_{j}^{i} \in \Delta(\Theta_{j} \times A_{j}) \text{ and } \text{marg}_{\Theta_{j}}\mu^{i}_{j} = p_{j} \text{ for each } j \neq i
\end{equation}
We use $\Delta^{p}_{i}$ to denote the set of all beliefs of $i$ satisfying (\ref{equ1}), and we let $\Delta^{p} = (\Delta^{p}_{j})_{j \in I}$. Here, the superscript $p$ emphasizes its crucial role in the restriction. 

Note that (\ref{equ1}) has two parts. First, $\mu^i$ is a product measure, which means that the belief of (an agent in) $i$ about one population is independent of her belief about another; therefore, it contains an independence assumption. Second, for each $j \neq i$, the marginal belief of each population $j$'s types should coincide with $p_j$. We call the property ``consistency with $p$'' instead of ``consistency with $p$ \emph{and independence}'' because, besides wanting to simplify the symbols, the independence is derived from a property of $p$, that is, $p$ is a profile $(p_j)_{j \in I} \in \times_{j \in I}\Delta(S_j)$ instead of (probably) a correlated distribution in $\Delta(\times_{j \in I}S_j)$. The condition (\ref{cm}) (consistency with monotonicity) in Section~\ref{sec:cm} should also be understood in the same manner.
\medskip

Restriction $\Delta^{p}_{i}$ is imposed on the first-order belief of agents from population $i$. Further, we assume that the restriction is \emph{transparent}, that is, the restriction holds and it is commonly believed to hold. Applying Battigalli and Siniscalchi's \cite{bs03} argument, the behavioral consequences of  rationality (R), common belief in rationality (CBR), and  transparency of $\Delta^{p}$ (TCP) are characterized by the iterative procedure defined as follows.

\begin{definition}\label{def2}
Consider the following procedure, called $\Delta^{p}$\textbf{-rationalization} \textbf{procedure}:

\textbf{Step 0}. For each $i \in I$, $\Sigma_{i,\Delta^{p}}^{0} = \Theta_{i} \times A_{i}$.

\textbf{Step $n+1$}. For each $i \in I$ and each $(\theta_{i},a_{i}) \in \Theta_{i} \times A_{i}$ with $(\theta_{i},a_{i}) \in \Sigma_{i,\Delta^{p}}^{n}$, $(\theta_{i},a_{i}) \in \Sigma_{i,\Delta^{p}}^{n+1}$ iff  there is some $\mu^{i} \in \Delta^{p}_{i}$ such that
\begin{enumerate}
\item $a_{i}$ is a best response to $\mu^{i}$ under $\theta_{i}$, and
\item $\mu^{i}(\Sigma_{-i,\Delta^{p}}^{n})=1$ where $\Sigma_{-i,\Delta^{p}}^{n} := \times_{j \neq i}\Sigma_{j,\Delta^{p}}^{n}$.
\end{enumerate}
Let $\Sigma_{i,\Delta^{p}}^{\infty} =\cap_{n\geq 0}\Sigma_{i,\Delta^{p}}^{n} $ for $i \in I$. The elements in $\Sigma_{i,\Delta^{p}}^{\infty}$ are said to be \textbf{$\Delta^{p}$-rationalizable}.
\end{definition}

A pair $(\theta_{i},a_{i})$ is called \emph{$n$-$\Delta^{p}$-rationalizable} iff $(\theta_{i},a_{i}) \in \Sigma_{i,\Delta^{p}}^{n}$. To survive at step $n$, $a_i$ has to be a best response under $\theta_i$ to a belief drawing upon the opponents' type-action pairs surviving up to the previous step and consistent with $p$, which shows the iterative structure of beliefs committing to rationality and restriction $\Delta^{p}$. A pair $(\theta_{i},a_{i})$ is $\Delta^{p}$-rationalizable iff it survives every step in the procedure---in other words, iff it is consistent with rationality,  common belief in rationality, and  transparency of $\Delta^{p}$.

Note that $\Delta^{p}$-rationalizability appeals to \emph{each player's} type-action pairs, while  an equilibrium is a \emph{profile} which specifies actions (for each Harsanyi type) \emph{for all players}. This explicitly embodies the characteristic methodology of EGT: appealing to decision-theoretic principles of rationality, that is, analyzing strategic reasoning based on \emph{individual} references within her own mind. In this sense, using the methodology of EGT to analyze the epistemic conditions for an equilibrium is regarded as providing a decision-theoretic foundation for the latter.\medskip

As mentioned in Section~\ref{sec:intro}, instead of $\Delta^{p}$-rationalizable type-action pairs or actions, our interest is in $\Delta^{p}$-rationalizable \emph{action distributions}, that is,  the action distributions obtained from $p$ and $\Delta^{p}$-rationalizable outcomes. Here is the formal definition.
\begin{definition}\label{norac}
We call $\pi_{i} \in \Delta(A_{i})$ an \textbf{$n$-$\Delta^{p}$-rationalizable action distribution} for $i$ iff there is $\mu_{i} \in \Delta(\Theta_{i} \times A_{i})$ with $\mu_{i}(\Sigma_{i,\Delta^{p}}^{n})=1$ and marg$_{\Theta_{i}}\mu_{i} = p_{i}$ such that $\pi_{i} =$ marg$_{A_{i}}\mu_{i}$. We call $\pi_{i}$ \textbf{$\Delta^{p}$-rationalizable} iff it is $n$-$\Delta^{p}$-rationalizable for each $n \in \mathbb{N}$.
\end{definition}

In words, $\pi_{i} \in \Delta(A_{i})$ is a ($n$-) $\Delta^{p}$-rationalizable action distribution if it can be obtained from a distribution $\mu_{i}$ on $\Delta(\Theta_{i} \times A_{i})$ whose support is composed of  $n$-$\Delta^{p}$-rationalizable type-action pairs. By requiring $\mu_{i}$ to be consistent with $p_{i}$, $\pi_{i}$ can essentially be pushed forward from $p_{i}$: one can consider that for each $i \in I$, there is a push-forward function $\sigma_{i}:\Theta_{i} \rightarrow A_{i}$ such that for each $\theta_{i} \in \Theta_{i}$, $(\theta_{i}, \sigma_{i}(\theta_{i})) \in \Sigma_{i,\Delta^{p}}^{\infty}$ ($\in \Sigma_{i,\Delta^{p}}^{n}$). Define $\pi_{i} \in \Delta(A_{i})$ such that $\pi_{i}(a_{i}) = p_{i}(\sigma_{i}^{-1}(a_{i}))$. It is clear that $\pi_{i}$ is ($n$-) $\Delta^{p}$-rationalizable.\footnote{We do not specify $\mu_{i}$ since any $\mu_{i}$ compatible with $\sigma_{i}$ works here. For example, $\mu_{i}$ can be defined as follows: for each measurable $E \in \Delta(\Theta_{i} \times A_{i}), \mu_{i}(E) := p_{i}($Proj$_{\Theta_{i}}E)$. Note that in the definition, the existence of some $\mu_{i} \in \Delta(\Theta_{i} \times A_{i})$ with $\mu_{i}(\Sigma_{i,\Delta^{p}}^{n})=1$ and marg$_{\Theta_{i}}\mu_{i} = p_{i}$ is guaranteed since $\Delta^{p}_{i}$ only imposes restrictions on beliefs about others' payoff types, which implies Proj$_{\Theta_{i}}\sum_{i,\Delta^{p}}^{n} = \Theta_{i}$ for each $n \in \mathbb{N}$. See Remark~3.2 in Battigalli and Siniscalchi \cite{bs03}.} By this notion, we can compare QRE (and other mixed-action equilibria) with $\Delta^{p}$-rationalizability.

Given a static game $G = \langle I, (A_{j},u_{j})_{j \in I}\rangle$ with $\langle\Theta,p\rangle$, we use $Q(G,\Theta, p)$ to denote the set of QREs. As noted in footnote \ref{fn1}, Brouwer's fixed-point theorem guarantees that $Q(G,\Theta, p) \neq \emptyset$. For each $i \in I$, $\pi_i \in \Delta(A_i)$ is said to be a \emph{QRE action distribution} (sometimes abbreviated as \emph{QRE distribution}) iff for some $\pi_{-i} \in \times_{j \neq i}\Delta(A_j)$, $(\pi_i, \pi_{-i}) \in Q(G,\Theta, p)$. First, we have the following result showing that each QRE action distribution is consistent with rationality, common belief in rationality, and transparency of $\Delta^{p}$.\footnote{As mentioned in Section~\ref{sec:intro}, Proposition~\ref{Thmai1} is conceptually straightforward. Yet in this special case one needs to define a proper push-forward function, which, to the best of our  knowledge, does not exist in the literature. Hence we still give the proof in \ref{sec:proof}.}

\begin{proposition}\label{Thmai1}
For each $\pi = (\pi_{j})_{j \in I} \in Q(G,\Theta, p)$ and each $i \in I$, $\pi_{i}$ is $\Delta^{p}$-rationalizable.
\end{proposition}

Examples in the next subsections will illustrate how the $\Delta^{p}$-rationalization procedure works. An algorithm to compute $n$-$\Delta^{p}$-rationalizable distributions for two-person games will be given there.

\subsection{Examples and an algorithm for two-person games}\label{sec:examp}

Here, we focus on two-person games and we use $i,j $ to denote different players, that is,  $i,j \in \{1,2\}$ and $i \neq j$.

\begin{example}\label{ex11}
Consider the (symmetric) Matching Pennies game in Table~\ref{TAB1}.\footnote{The games in Examples~\ref{ex11} and  \ref{ex33} were experimentally studied in Ochs \cite{oc95}, from which we adopt the term \textit{(a)symmetric Matching Pennies game}.} For each $i = 1,2$, $p_{i}$ is the uniform distribution on $[-2,2] \times [-2,2]$. The unique QRE is $\pi^{*}=((\frac{1}{2}H,\frac{1}{2}T),(\frac{1}{2}H,\frac{1}{2}T))$, that is, for each player, the two actions are used randomly equally. We show that, for each $i = 1,2$, the only $\Delta^{p}$-rationalizable distribution is $\pi^{*}_{i}$.  

\begin{table}[ht!]

\centering
 \begin{tabular}{c c c }

 \hline
 $1\setminus 2$ & $H$ & $T$ \\ 
 \hline
$H$ & $1,0$ &  $0,1$ \\ 

 $T$ & $0,1$ & $1,0$ \\

 \hline
\end{tabular}
\caption{Example \ref{ex11}}
\label{TAB1}

\end{table}

\begin{figure}[ht]
\centering
  \includegraphics[width=1\columnwidth]{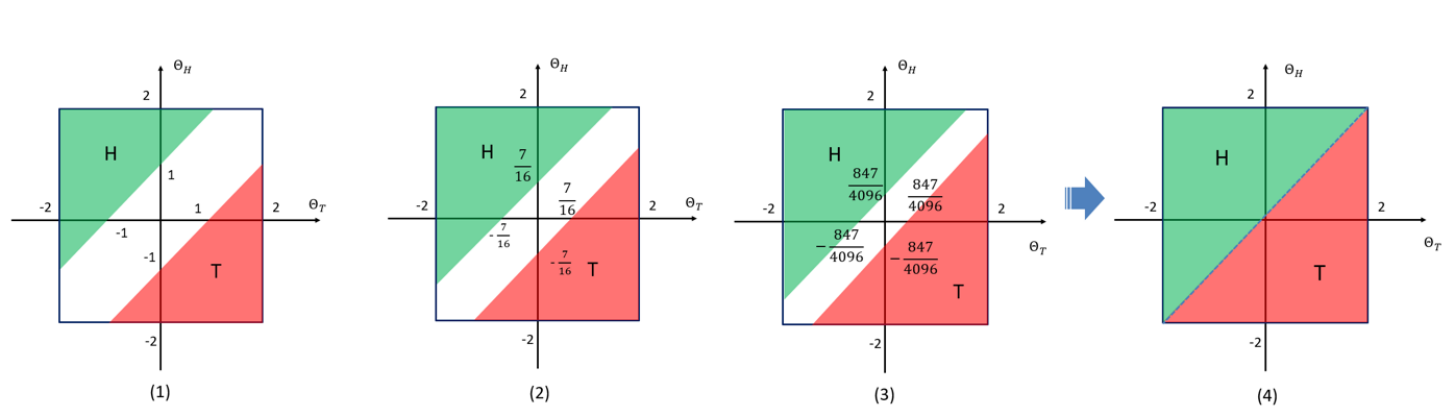}
  \caption{The $\Delta^{p}$-rationalization procedure converges to the QRE}
  \label{fig:MES}
\end{figure}

Let $i \in I$ and $\theta_{i} =(\theta_{iH},\theta_{iT}) \in \Theta_{i}$. First, $(\theta_{i},H) \in \Sigma_{i,\Delta^{p}}^{1}$ if and only if for some probability $q_{j} = (q_{j}(H), q_{j}(T)) \in \Delta(A_{j})$, $u_{i}(H,q_{j})+\theta_{iH} \geq u_{i}(T,q_{j})+\theta_{iT}$, or, by rearranging the inequality,
\begin{equation}\label{rty10}
\theta_{iH}-\theta_{iT} \geq q_{j}(H)[u_{i}(T,H) - u_{i}(H,H)] + q_{j}(T)[u_{i}(T,T) - u_{i}(H,T)]
\end{equation}
Note that if $\theta_{iH}-\theta_{iT}$ is bigger that the maximum value of the left-hand side of inequality (\ref{rty10}), obviously $H$ is a best response. Since $u_{i}(T,H) - u_{i}(H,H) = -1$ and $u_{i}(T,T) - u_{i}(H,T)=1$, to obtain the maximum, we only need to maximize the coefficient of $[u_{i}(T,T) - u_{i}(H,T)]$ and minimize that of $[u_{i}(T,H) - u_{i}(H,H)]$, that is, to let $q_{j}(H) = 0$ and $q_{j}(T) = 1$.\footnote{This $q_{j}$ can be generated from a belief consistent with $p$. Let $\Theta_{j}(T) =  \Theta_{j} \times \{T\}$ and define $\mu^{i}$ such that for each measurable $E \subseteq  \Theta_{j} \times A_{j}$, $\mu^{i}(E) = p_{j}($Proj$_{\Theta_{j}}(E \cap \Theta_{j}(T)))$. One can easily see that (i) $\mu^{i} \in \Delta^{p}_{i}$ and (ii) marg$_{A_{j}}\mu^{i} = 0H+1T$. Similar results can be proved for each step in the procedure.} Hence, when $\theta_{iH}-\theta_{iT} \geq 1$ (the green area in Figure~\ref{fig:MES} (1)), $(\theta_{i},H) \in \Sigma_{i,\Delta^{p}}^{1}$. Similarly, when $\theta_{iT}-\theta_{iH} \geq 1$ (the red area in Figure~\ref{fig:MES} (1)), $(\theta_{i},T) \in \Sigma_{i,\Delta^{p}}^{1}$. To every type between them (that is,  the white area) both $H$ and $T$ could be associated: each could be supported by some allowable beliefs.

Note that the probability measure of each area is $\frac{9}{32}$. This  means that for each $\mu^{i} \in \Delta^{p}_{i}$ with $\mu^{i}(\Sigma_{j,\Delta^{p}}^{1})=1$, both marg$_{A_{j}}\mu^{i}(H)$ and marg$_{A_{j}}\mu^{i}(T)$ are at least $\frac{9}{32}$. This imposes a new restriction on  allowable  beliefs. Therefore,  at the second step of the procedure, to obtain the maximum of the left-hand side of inequality (\ref{rty10}) , we need to let $q_{j}(H)$ at least $\frac{9}{32}$ and $q_{j}(T)$ at most $1- \frac{9}{32} = \frac{23}{32}$, which implies that now the maximum value is $\frac{7}{16}$; in other words, at the second step, when $\theta_{iH}-\theta_{iT} \geq  \frac{7}{16}$, $(\theta_{i}, H) \in \Sigma_{i,\Delta^{p}}^{2}$; similarly, when $\theta_{iT}-\theta_{iH} \geq  \frac{7}{16}$,  $(\theta_{i},T) \in \Sigma_{i,\Delta^{p}}^{2}$. Those are illustrated as the green and red areas in Figure~\ref{fig:MES} (2), respectively. Now we obtain restrictions on for the beliefs in the third step, which leads to Figure~\ref{fig:MES} (3).

We can continue this process: at step $n$, by replacing $q_{j}(H)$ with the lower bound $q_{j,H}^{n}$ according to the restriction imposed at the previous step, and by letting $q_{j}(T)$ be $1-q_{j,H}^{n}$, we obtain via inequality (\ref{rty10}) a new condition of restrictions on the lower bound of the probability of each action in beliefs. In general, let $q^{n}_{i,a_{i}}$ be the lower bound of the probability of $a_{i} \in \{H,T\}$ in beliefs at step $n$. Then one can see that $q^{0}_{i,a_{i}} = 0$ and for $n \geq 0$,
\begin{equation*}
q^{n+1}_{i,a_{i}}=\frac{(3+2q^{n}_{i,a_{i}})^{2}}{32},
\end{equation*}
which is bounded and increasing; hence $\lim_{n \rightarrow \infty}q^{n}_{i,a_{i}} =  \frac{1}{2}$, the QRE action distribution. $\blacktriangle$

\end{example}

The calculation in Example~\ref{ex11} shows that for each $a_{i}$, the set of $\theta_{i}$ for which $a_{i}$ is ``absolutely'' preferred to all other actions is important since its measure with respect to $p_{i}$ provides a lower bound for the probability of $a_{i}$ in beliefs at the next step. Then what conditions characterize the set? We decompose this question into smaller modulars: for $a_{i}^{\prime} \in A_{i}$ ($a_{i}^{\prime} \neq a_{i}$), when $u_{i}(a_{i},q_{j})+\theta_{i,a_{i}} \geq u_{i}(a_{i}^{\prime},q_{j})+\theta_{i,a_{i}^{\prime}}$ holds for each allowable $q_{j}$? Since, for each $q_{j} \in \Delta(A_{j})$, $u_{i}(a_{i}^{\prime},q_{j}) - u_{i}(a_{i},q_{j}) = \sum_{a_{j} \in A_{j}} q_{j}(a_{j})[u_{i}(a_{i}^{\prime},a_{j}) - u_{i}(a_{i},a_{j})]$, it is equivalent to asking when the following condition holds:
\begin{equation}\label{conde}
\tag{\textbf{MAX}}
\theta_{i,a_{i}}  - \theta_{i,a_{i}^{\prime}} \geq\sum_{a_{j} \in A_{j}} q_{j}(a_{j})[u_{i}(a_{i}^{\prime},a_{j}) - u_{i}(a_{i},a_{j})]
\end{equation}
To satisfy the condition (\ref{conde}), one only has to consider what the maximum of $\sum_{a_{j} \in A_{j}} q_{j}(a_{j})$ $(u_{i}(a_{i}^{\prime},a_{j}) - u_{i}(a_{i},a_{j}))$ is over all allowable $q_{j}$ at each step. As shown in Example~\ref{ex11}, this allowability is recursively defined in the procedure. As one can see from basic linear optimization theory, it is helpful to decompose $A_{j}$ into two groups:
\begin{subequations} \label{all}
    \begin{align}
   \Phi_{i}(a_{i},a_{i}^{\prime}) & :=\{a_{j} \in A_{j}: u_{i}(a_{i}^{\prime},a_{j}) - u_{i}(a_{i},a_{j}) \geq u_{i}(a_{i}^{\prime},a^{\prime}_{j}) - u_{i}(a_{i},a^{\prime}_{j}) \text{ for all } a^{\prime}_{j} \in A_{j}\} \label{Big ones}\\
    \phi_{i}(a_{i},a_{i}^{\prime}) & :=A_{j} \setminus \Phi_{i}(a_{i},a_{i}^{\prime}) \label{Small ones}
     \end{align}
\end{subequations}
In words, $\Phi_{i}(a_{i},a_{i}^{\prime})$ is the set of the opponent's actions that maximize the improvement of $i$'s payoff if she deviates from $a_{i}$ to $a_{i}^{\prime}$, and $\phi_{i}(a_{i},a_{i}^{\prime})$ is its complement. At each step $n$, by making $\sum_{a_{j} \in \Phi_{j}(a_{i},a_{i}^{\prime})}q_{j}(a_{j})$ as large as possible and the weighted payoff over $ \phi_{i}(a_{i},a_{i}^{\prime})$ as small as allowable according to the condition imposed in the previous step, we obtained the set of $\theta_{i}$ satisfying condition (\ref{conde}), which brings about the measure of the set of payoff types at step $n$ of the $\Delta^{p}$-rationalization procedure under which $a_{i}$ is chosen regardless of the belief about others' behavior. We denote this measure by $q_{i,a_{i}}^{n}$, which provides a lower bound for allowable $q_{i}(a_{i})$ at step $n+1$.

The above argument is summarized in Algorithm~\ref{alg:cap}. Here, we let $\overline{H}_{i}(a_{i},a_{i}^{\prime}) =  u_{i}(a_{i}^{\prime},a_{j}) - u_{i}(a_{i},a_{j})$ for an $a_{j} \in \Phi_{i}(a_{i},a_{i}^{\prime})$. By the definition of $\Phi_{i}(a_{i},a_{i}^{\prime})$, $\overline{H}_{i}(a_{i},a_{i}^{\prime})$ is well defined. The algorithm is written in a pseudocode that can be translated into any programming language. Given a two-person game and $p = (p_{j})_{j \in I}$, the algorithm gives  $q_{i,a_{i}}^{n}$ for each $i$ and $a_{i} \in A_{i}$ up to any step $n$.  

\begin{algorithm}[h!]
\caption{An algorithm for computing $(q_{i,a_{i}}^{n})_{i \in I, a_{i} \in A_{i}}$, $n \in \mathbb{N}$}\label{alg:cap}
\begin{algorithmic}
\State \textbf{start:} $q_{i,a_{i}}^{o} =0$ for each $i \in I$ and $a_{i} \in A_{i}$
\While{$t < n$}
\State Set $V_{i}^{t+1}(a_{i},a_{i}^{\prime}) = \sum_{a_{j} \in \phi_{i}(a_{i},a_{i}^{\prime})}u_{i}(a_{i},a_{j})q_{j,a_{j}}^{t} + \overline{H}_{i}(a_{i},a_{i}^{\prime})(1- \sum_{a_{j} \in \phi_{i}(a_{i},a_{i}^{\prime})}q_{j,a_{j}}^{t})$ 
\State Set $E_{i,a_{i}}^{t+1} = \{\theta_{i}\in \mathbb{R}^{A_{i}}: \theta_{i,a_{i}}-\theta_{i,a_{i}^{\prime}} \geq V_{i}^{t+1}(a_{i},a_{i}^{\prime}) \text{ for each }a_{i}^{\prime} \in A_{i} \setminus \{a_{i}\}\}$
\State Set $q_{i,a_{i}}^{t+1}=p_{i}(E_{i,a_{i}}^{t+1})$
\EndWhile

\end{algorithmic}
\end{algorithm}
In Algorithm~\ref{alg:cap}, $V_{i}^{t+1}(a_{i},a_{i}^{\prime})$ is the maximum of the gap between the payoff from $a_{i}^{\prime}$ to $a_{i}$ based on all allowable distributions at step $t$. As we argued before, it is obtained by letting the probability of actions in  $\phi_{i}(a_{i},a_{i}^{\prime})$ be as small as possible (that is, letting them be $q_{j,a_{j}}^{t}$, the measure of the set of types that absolutely prefer $a_{i}$ to $a_{i}^{\prime}$ allowed at step $t$) and letting the actions in $\Phi_{i}(a_{i},a_{i}^{\prime})$ be as large as possible (that is, letting  their sum be $(1- \sum_{a_{j} \in \phi_{i}(a_{i},a_{i}^{\prime})}q_{j,a_{j}}^{t}$)). Based on this, $E_{i,a_{i}}^{t+1}$ is the set of types under which $a_i$ is preferred to any other actions to any allowable beliefs at step $t+1$, and $q_{i,a_{i}}^{t+1}$ is its measure, which provides the recursive base of the computation in the next step.\footnote{Since $E_{i}^{t+1}(a_{i})$ is a polyhedron defined by linear inequalities, it is Borelian and hence  measurable.}

Because of  its critical role in Algorithm~\ref{alg:cap}, the decomposition of $A_{j}$ into $\Phi_{i}(a_{i},a_{i}^{\prime})$ and $\phi_{i}(a_{i},a_{i}^{\prime})$, particularly the latter, will become significant when we study the coincidence of QRE and $\Delta^{p}$-rationalizable action distributions in Section~\ref{sec:uncer}. \medskip

The following statement shows that the boundary imposed by each $q_{i,a_{i}}^{n}$ is tight.
\begin{proposition}\label{prop2}
For each $n \in \mathbb{N}$, $i = 1, 2$ and $a_{i} \in A_{i}$, $q_{i,a_{i}}^{n}$ is the minimum of probabilities that $a_{i}$ is used in all $n$-$\Delta^{p}$-rationalizable distributions, that is, $p_{i}(a_{i}) \geq q_{i,a_{i}}^{n}$ for each $n$-$\Delta^{p}$-rationalizable distribution $p_{i}$ and no number bigger than $q_{i,a_{i}}^{n}$ satisfies this property.
\end{proposition}

\begin{example}\label{exam:vacc}
Consider the asymmetric game of chicken in Table~\ref{TABIntro} (Goeree et al. \cite{ghp16}, pp.~25-26). Suppose that $p_{i}$'s are i.i.d. and each has an extreme value distribution with parameter $\lambda=0.5$, that is, the c.d.f. $F_{i}(\theta_{i,a_{i}}) = \exp(-\exp (-0.5 \theta_{i,a_{i}}))$ for each $i \in \{1,2\}$ and $a_{i} \in \{T,S\}$. Using Algorithm~\ref{alg:cap}, the sequences of minimums are depicted in Figure~\ref{fig:Conv0}. One can see that the sequences converge to the unique QRE. 
\begin{table}[ht]

\centering
 \begin{tabular}{c c c }

 \hline
 $1\setminus 2$ & $T$ & $S$ \\ 
 \hline
 $T $& $0, 0$ &  $6, 1$ \\ 

 $S$ & $1, 14$ & $ 2, 2$ \\

 \hline
\end{tabular}
\caption{An asymmetric game}
\label{TABIntro}

\end{table}

\begin{figure}[ht]
\centering
  \includegraphics[width=0.9\columnwidth]{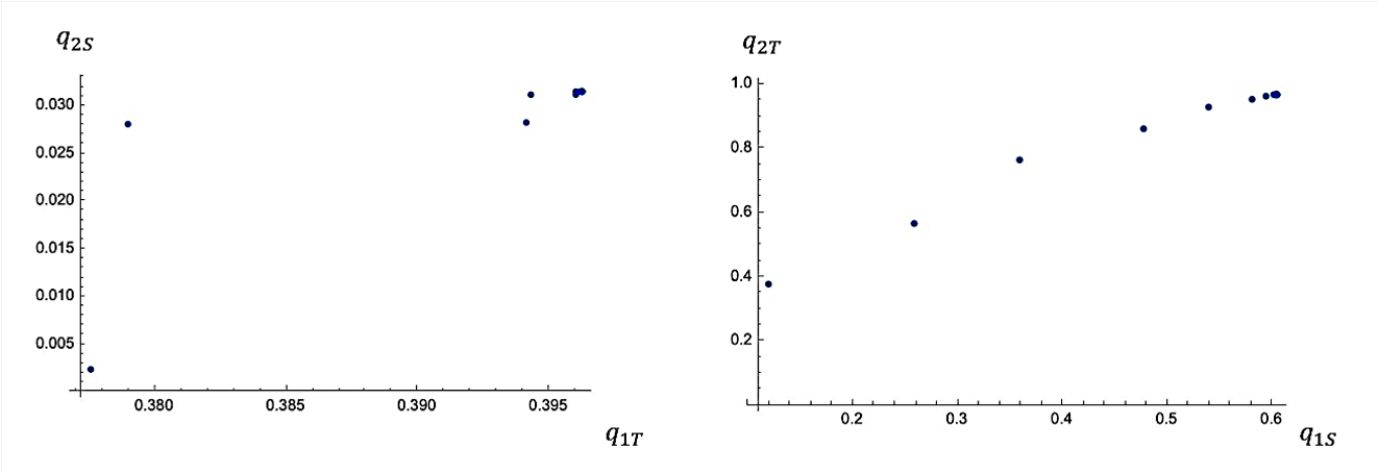}
  \caption{The convergence of minimums for each action}
  \label{fig:Conv0}
\end{figure}

 Goeree et al. \cite{ghp16} (Figure~2.4, p.~25) showed that as $\lambda$ increases, the number of QREs increases from one to multiple. One may wonder what would happen to the minimum sequences. By applying Algorithm~\ref{alg:cap}, one can see that when there are multiple QREs, for each $a_{i}$, the sequence $(q_{i,a_{i}}^{n})_{n=1}^{\infty}$ converges to the smallest proportion of $a_{i}$ in all QRE distributions. Still, QRE provides ``tight'' information for the $\Delta^{p}$-rationalizable results.  $\blacktriangle$
 
 \end{example}
 
So far, we have seen only examples in which each $(q_{i, a_i}^n)_{n=1}{\infty}$ converges to the proportion of $a_i$ used in some QRE. However, as the following example  shows, it is not always true. In general, QRE provide sharper predictions than $\Delta^p$-rationalizability.
 
\begin{example}\label{ex33}
Consider the asymmetric Matching Pennies game in Table~\ref{TAB4}. Here, for each $i =1,2$, $p_i=\prod_{k=H,T} p_{ik}$, where each $p_{ik}$ is the extreme value distribution with $\lambda=5$. This game has only one QRE (Goeree et al. \cite{ghp16}, Chapter 2.2).
\begin{table}[ht]

\centering
 \begin{tabular}{c c c }

 \hline
$1\setminus 2$ & $H$ & $T$ \\ 
 \hline
$H$ & $9,0$ &  $0,1$  \\ 

 $T$ & $0,1$ &  $1,0$ \\

 \hline
\end{tabular}
\caption{Example \ref{ex33}}
\label{TAB4}
\end{table}

It can be shown that for each $i = 1,2, a_{i} = H,T$, the infimum sequence $(q_{i,a_{i}}^{n})_{n=0}^{\infty}$ converges to some number extremely close to $0$, which means that their limits do not form the QRE since the sum of each pair is strictly less than 1.

\end{example}

\section{When \texorpdfstring{$\Delta^{p}$}{delta^p}-rationalizability implies QRE action distributions}\label{sec:uncer}

\noindent  Example~\ref{ex33} shows that the converse of Proposition~\ref{Thmai1} does not hold. It is  common  that the equilibria under examination form only a proper subset of the behavioral consequence consistent with its epistemic conditions. This is largely because of the decision-theoretic nature of EGT, which permits beliefs that, while possibly  incorrect  ex post, remain in line with the specified epistemic conditions. Nevertheless, Examples~\ref{ex11} and~\ref{exam:vacc} suggest that under some conditions, the converse of Proposition~\ref{Thmai1} could be true, that is, QRE provides  tight  (when there is a single QRE) or boundary (when there are multiple QREs) information about the $\Delta^{p}$-rationalizable action distributions. As pointed out in Example~\ref{exam:vacc}, we compare the following two objects: (i) the minimum proportion of each action in all QRE distributions, and (ii) the minimum proportion of each action in $n$th-order $\Delta^{p}$-rationalizable distributions. For (i), formally, for each $i \in I$ and $a_{i} \in A_{i}$, we define 
\begin{equation}\label{eq:lob}
\underline{\pi}_{i,a_{i}}=\min \{\pi_{i}(a_{i}): (\pi_{i}, \pi_{-i}) \in Q(G,\Theta,p) \text{ for some } \pi_{-i} \in \Delta(A_{-i})\},
\end{equation}
Because of the assumption about $p_{i}$ of continuity and existence of marginal density function for each $a_{i} \in A_{i}$, $\underline{\pi}_{i,a_{i}}$ is well defined.

For (ii), though it is already defined recursively as $q_{i,a_{i}}^{n}$ in Algorithm~\ref{alg:cap} in two-person games, we give it a general and explicit definition here, which could facilitate proofs. For each $n \in \mathbb{N}$, let
\begin{equation}
q_{i,a_{i}}^{n} = \min\{\pi_{i}(a_{i}): \pi_{i}  \text{ is } n\text{-}\Delta^{p}\text{-rationalizable}\}
\end{equation}
It is straightforward to see that the two definitions of $q_{i,a_{i}}^{n}$ coincide in two-person games.\medskip

First, the following result shows that the minimum sequence  $(q_{i,a_{i}}^{n})_{n =0}^{\infty}$ is convergent for each $i \in I$ and $a_{i} \in A_{i}$.

\begin{proposition}\label{PropDec}
For each $i \in I$ and $a_{i} \in A_{i}$, $(q_{i,a_{i}}^{n})_{n =0}^{\infty}$ is bounded and nondecreasing. Hence, there is some $q_{i,a_{i}}^{*}$ such that $\lim_{n \rightarrow \infty} q_{i,a_{i}}^{n} =q_{i,a_{i}}^{*}$. In addition, $q_{i,a_{i}}^{*} \leq \underline{\pi}_{i,a_{i}}$. 
\end{proposition}

Drawing upon Proposition~\ref{PropDec}, the next step is to find the condition for $q_{i,a_{i}}^{*} = \underline{\pi}_{i,a_{i}}$. Before proceeding, it is imperative to specify the nature of this condition: whether it pertains to (additional) epistemic characteristics, the properties about $p$, or aspects relevant to the structure of the game. First, it appears that transparency of $\Delta^{p}$ already comprehensively addresses the information about beliefs encapsulated within QRE. Second, certain probability distributions $p_{i}$, especially those frequently employed in the literature, such as extreme value distributions, have the potential to introduce substantial constraints. These constraints might encompass a wealth of relevant information concerning the possibility of the coincidence. However, our focus is drawn to the intriguing observation that in certain games, as exemplified in Example~\ref{exam:vacc}, $q_{i,a_{i}}^{*} = \underline{\pi}_{i,a_{i}}$ holds regardless of the underlying probability distributions of idiosyncrasies. Therefore, we focus on properties of the structure of the game.\footnote{We are not the first to investigate the condition on the properties of the structure of a game for the coincidence of an equilibrium and a solution concept characterizing the behavioral consequence of some epistemic conditions. For example, Milgrom and Roberts \cite{mr90} studied the crucial role of supermodularity for the coincidence of rationalizability and Nash equilibrium. }

Our purpose now is to find conditions about the game under which $q_{i,a_{i}}^{*}=\underline{\pi}_{i,a_{i}}$, or, equivalently, the following statement, holds.
\begin{statement}\label{st1}
For each $i \in I$ and $a_{i} \in A_{i}$, $q_{i,a_{i}}^{*}$ represents  the probability that action $a_{i}$ is used in some QRE.
\end{statement}
Indeed, since $\underline{\pi}_{i,a_{i}}$ is the minimum proportion of action $a_{i}$ in all QRE, if $q_{i,a_{i}}^{*}$ is used in some QRE, it follows that $q_{i,a_{i}}^{*} \geq \underline{\pi}_{i,a_{i}}$. Combining this with Proposition~\ref{PropDec}, we obtain $q_{i,a_{i}}^{*} = \underline{\pi}_{i,a_{i}}$. Before presenting the formal condition and result, we use a heuristic example to facilitate a better understanding of the underlying idea.

\begin{example}\label{exa:heu}
Given the $2 \times 2$ game in Table~\ref{TAB5} with $p = (p_{1},p_{2})$, we consider the two cases in which Statement~\ref{st1} does and does not hold, respectively. 

\begin{table}[ht]
\centering
 \begin{tabular}{c c c }
\hline
 & $L$ & $R$ \\ 
 \hline
$T$ & $\alpha_{1}, \alpha_{2}$ &  $\beta_{1},\beta_{2}$  \\ 
$U$ & $\gamma_{1},\gamma_{2}$ &  $\delta_{1},\delta_{2}$ \\ 
 \hline
\end{tabular}
\caption{A $2 \times 2$ game}
\label{TAB5}
\end{table}

\textbf{Case 1 (Statement \ref{st1} holds).} Recall that in Section~\ref{sec:examp}, we defined $\Phi_{i}(a_{i},a_{i}^{\prime})$ to be the set of the opponent's actions that maximize the improvement of $i$'s payoff if she deviates from $a_{i}$ to $a_{i}^{\prime}$, and $\phi_{i}(a_{i},a_{i}^{\prime})$ is its complement (see formulas~(\ref{all}a) and (\ref{all}b)); also, we defined $\overline{H}_{i}(a_{i},a_{i}^{\prime})$ to be the amount of improvement, that is, $\overline{H}_{i}(a_{i},a_{i}^{\prime})  =  u_{i}(a_{i}^{\prime},a_{j}) - u_{i}(a_{i},a_{j})$ for an $a_{j} \in \Phi_{i}(a_{i},a_{i}^{\prime})$. Here, we assume that $\overline{H}_{1}(T,U) = \gamma_{1}-\alpha_{1} > \delta_{1}-\beta_{1}=:\underline{H}_{1}(T,U) $ and $\overline{H}_{2}(L,R) = \beta_{2}-\alpha_{2} > \delta_{2}-\gamma_{2}=: \underline{H}_{2}(L,R)$. By this assumption, it is easy to see that $\phi_{1}(T,U) = \{R\}$, $\phi_{1}(U,T) = \{L\}$, $\phi_{2}(L,R) = \{U\}$, and $\phi_{2}(R,L) = \{T\}$.

The iterative procedure in Algorithm~\ref{alg:cap} shows that, in a $2 \times 2$ game, at each step $n$, the value $q_{i,a_{i}}^{n}$ completely relies upon $q_{j,a_{j}}^{n-1}$ for $a_{j} \in \phi_{i}(a_{i},a_{i}^{\prime})$: by replacing $q_{j,a_{j}}$ with $q_{j,a_{j}}^{n-1}$ for $a_{j} \in \phi_{i}(a_{i},a_{i}^{\prime})$ and $q_{j,a^{\prime}_{j}}$ with $(1-q_{j,a_{j}}^{n-1})$ for $a^{\prime}_{j} \in \Phi_{i}(a_{i},a_{i}^{\prime})$ in the right-hand side of (\ref{conde}) , one obtains the area in $\Theta_{i}$ for which $a_{i}$ is absolutely better than $a_{i}^{\prime}$, whose measure is $q_{i,a_{i}}^{n}$. By continuity, in the limit we have
\begin{equation}\label{poi1}
q_{1T}^{*} =p_{1}(\{\theta_{1}: \theta_{1T}-\theta_{1U} \geq (1-q_{2R}^{*})\overline{H}_{1}(T,U) + q_{2R}^{*}\underline{H}_{1}(T,U)\})= p_{1}(E_{1T}(1-q_{2R}^{*}, q_{2R}^{*}))
\end{equation}
\begin{equation}
 q_{1U}^{*} =p_{1}(\{\theta_{1}:\theta_{1T}-\theta_{1U} \leq q_{2L}^{*}\overline{H}_{1}(T,U) + (1-q_{2L}^{*})\underline{H}_{1}(T,U)\})= p_{1}(E_{1U}(q_{2L}^{*}, 1-q_{2L}^{*}))
\end{equation}
\begin{equation}
q_{2L}^{*} =p_{2}(\{\theta_{2}:\theta_{2L}-\theta_{2R} \geq (1-q_{1U}^{*})\overline{H}_{2}(L,R) + q_{1U}^{*}\underline{H}_{2}(L,R)\})= p_{2}(E_{2L}(1-q_{1U}^{*}, q_{1U}^{*}))
\end{equation}
\begin{equation}\label{oyi1}
 q_{2R}^{*} =p_{2}(\{\theta_{2}:\theta_{2L}-\theta_{2R} \leq q_{1T}^{*}\overline{H}_{2}(L,R) + (1-q_{1T}^{*})\underline{H}_{2}(L,R)\})= p_{2}(E_{2R}(q^{*}_{1T},1-q^{*}_{1T}))
\end{equation}

Is there any QRE in which player 1 uses $T$ with probability $q_{1T}^{*}$? The answer is yes. Consider $(q_{1},q_{2}) \in \Delta(A_{1}) \times \Delta(A_{2})$ such that
\begin{displaymath}
q_{1}(T) = q_{1T}^{*},{ \ \ } q_{1}(U)=1- q_{1T}^{*}, { \ \ } q_{2}(L) = 1-q_{2R}^{*},{ \ \ }  q_{2}(R)=q_{2R}^{*}.
\end{displaymath}
We show that $(q_{1},q_{2})$ is a QRE.  First, by (\ref{poi1}) and (\ref{oyi1}), it is clear that $q_{1}(T) = q_{1T}^{*} = p_{1}(E_{1T}(1-q_{2R}^{*}, q_{2R}^{*}))$ and $q_{2}(R) = q_{2R}^{*}=p_{2}(E_{2R}(q^{*}_{1T},1-q^{*}_{1T}))$. For action $U$, we have 
\begin{equation*}
\begin{split}
& p_{1}(E_{1U}(1-q_{2R}^{*}, q_{2R}^{*}) ) =p_{1}(\{\theta_{1}: \theta_{1T}-\theta_{1U} \leq (1-q_{2R}^{*})\overline{H}_{1}(T,U)+q_{2R}^{*}\underline{H}_{1}(T,U)\}) \\
& = p_{1}(\Theta_{1} \setminus \{\theta_{1}: \theta_{1T}-\theta_{1U} > (1-q_{2R}^{*})\overline{H}_{1}(T,U)+q_{2R}^{*}\underline{H}_{1}(T,U)\})\\
& = 1 - p_{1}(E_{1T}(1-q_{2R}^{*}, q_{2R}^{*})) \\
& = 1- q_{1T}^{*}
\end{split}
\end{equation*}
Therefore, $p_{1}(E_{1U}(1-q_{2R}^{*}, q_{2R}^{*})) = q_{1}(U)$.\footnote{The third equality in the calculation holds because $p_{i}$, as we assumed, is absolutely continuous.} Similarly, we can show that $p_{2}(E_{2L}(q_{1T}^{*}, 1-q_{1T}^{*}) )= q_{2}(L)$. By Definition~\ref{def1}, $q$ is indeed a QRE. Note that $q$ is also a QRE in which $R$ is used by $q_{2R}^{*}$. In a similar manner, we can build a QRE in which $U$ is used by $q_{1U}^{*}$ and $L$ is used by $q_{2L}^{*}$. Therefore, in this case, $q_{i,a_{i}}^{*} = \underline{\pi}_{i,a_{i}}$ for each $i\in I$ and $a_{i} \in A_{i}$. In particular, when QRE is unique, it is the only $\Delta^{p}$-rationalizable distribution profile.\medskip

The key for the coincidence here is a separated interdependence between actions of the two players: $(q_{1,T}^{n})_{n =0}^{\infty}$ and $(q_{2,R}^{n})_{n =0}^{\infty}$ are interdependent, and $(q_{1,U}^{n})_{n =0}^{\infty}$ and $(q_{2,L}^{n})_{n =0}^{\infty}$ are interdependent. Importantly, these two groups of interdependent actions are isolated from each other. Under this condition, for each action $a_{i}$ (and its interdependent partner $a_{j}$) we can implement $q_{i,a_{i}}^{*}$ (and $q_{j,a_{j}}^{*}$) with a QRE by adjusting the shares of the ``irrelevant'' actions. In contrast, when this separated interdependence is not satisfied, the coincidence may not hold.\medskip

\textbf{Case 2 (Statement \ref{st1} does not hold).} Assume that $\overline{H}_{1}(T,U) =  \delta_{1}-\beta_{1} > \gamma_{1}-\alpha_{1} $ and, still, $\overline{H}_{2}(L,R)= \beta_{2}-\alpha_{2} >\delta_{2}-\gamma_{2}$. In the vein of the argument above, one can see that $q_{1T}^{*}$ relies upon $q_{2L}^{*}$, $q_{2L}^{*}$ relies upon $q_{1U}^{*}$, $q_{1U}^{*}$ relies upon $q_{2R}^{*}$, and $q_{2R}^{*}$ relies upon $q_{1T}^{*}$. In other words, the interdependence traverses over all actions and $A_{1} \cup A_{2}$ is not divided into two interdependent groups.

Is there any QRE in which player 1 uses $T$ with probability $q_{1T}^{*}$? Sometimes. When $q_{2L}^{*}+q_{2R}^{*} =1$ (which implies that $q_{1T}^{*}+q_{1U}^{*} =1$), the answer is yes. However, when $q_{2L}^{*}+q_{2R}^{*} <1$ (which implies that $q_{1T}^{*}+q_{1U}^{*} <1$), the answer is definitely no. To see this, suppose that in some QRE $q=(q_{1},q_{2})$, player 1 uses $T$ with probability $q_{1}(T)=q_{1T}^{*}$. Using an argument similar to that in Case~1, $q_{2}(L)=q_{2L}^{*}$ implies that $q_{1}(U)=q_{1U}^{*}$, and consequently it leads to $q_{2}(R) = q_{2R}^{*}$. However, by assumption, $p_{2}(L)+p_{2}(R)=q_{2L}^{*}+q_{2R}^{*} <1$, which means that $q_{2}$ is not even a probability distribution. \medskip

Case~2 illustrates why the  separated interdependence  of actions is relevant. When all actions are connected by the interdependence relation, fixing one action's share to be $q_{i,a_{i}}^{*}$ means that the same has to be done to every action; in other words, the whole system is  inflexible  and there is no space to adjust the probabilities of  irrelevant  actions to form a QRE. It does not cause any problem when $\sum_{a_{i} \in A_{i}}q_{i,a_{i}}^{*}=1$ for each $i$ (as in Example~\ref{ex11}, which is rare); yet the coincidence does not hold when the sum is smaller than 1, for example, when there are multiple QREs. $\blacktriangle$
\end{example}

We now formalize the idea. Consider a two-person game $G=\langle A_{1},A_{2},u_{1},u_{2}\rangle$. Recall that, as defined in Section~\ref{sec:examp}, $\phi_{i}(a_{i},a_{i}^{\prime}) =\{a_{j} \in A_{j}:u_{i}(a_{i}^{\prime},a_{j})-u_{i}(a_{i},a_{j}) < \overline{H}_{i}(a_{i},a_{i}^{\prime})\}$. Based on $\phi_{i}$, we define a correspondence $\tilde{\phi}_{i}:A_{i} \rightarrow 2^{A_{j}}$ such that $\tilde{\phi}_{i}(a_{i})=\cup_{a_{i}^{\prime}\neq  a_{i}}\phi_{i}(a_{i},a_{i}^{\prime})$ for each $i \in I$ and $a_{i} \in A_{i}$. Each $a_{j} \in \tilde{\phi}_{i}(a_{i})$ is called an \emph{influential action} for $a_{i}$ since the minimums of the probabilities of those actions determine the area in $\Theta_{i}$ in which $a_{i}$ is a best response. An action $a_{i} \in A_{i}$ is called \emph{\textbf{non-serial}} iff $\tilde{\phi}_{i}(a_{i})= \emptyset$; it is called \textbf{\emph{eventually non-serial}} iff it is non-serial or for some $a_{j} \in \tilde{\phi}_{i}(a_{i})$, $\tilde{\phi}_{j}(a_{j} )= \emptyset$.\footnote{The names are adopted from modal logic. See, for example, Blackburn, De Rijke, and Venema \cite{bdv01}.} Note that one action's  non-seriality implies that each action of the player is non-serial; consequently, non-seriality indicates degenerate cases. A game is called \emph{serial} iff no action is non-serial.

We can visualize $(\tilde{\phi}_{1},\tilde{\phi}_{2})$ with a directed graph $(\Rightarrow, A_{1}\cup A_{2})$ such that for each $a_{i},a_{j} \in A_{1} \cup A_{2}$, $a_{i} \Rightarrow a_{j}$ iff $a_{j} \in \tilde{\phi}_{i}(a_{i})$. For example, Figure~\ref{fig:DIR} (1) is the graph of the game in Table~\ref{TABIntro}, and Figure~\ref{fig:DIR} (2) is that of Example~\ref{ex11}. Intuitively, one can see that the former satisfies the separated-interdependence  condition mentioned in Example~\ref{exa:heu}, while the latter does not. Theorem~\ref{Thmai2} formalizes the heuristic idea. Note that, though the condition can be defined for all two-person games,  in Theorem~\ref{Thmai2} we only give the form for $2 \times 2$ games and leave the generalization to \ref{appen2}. The rationale lies in two considerations: first, the condition has a more straightforward representation in $2 \times 2$ games; second, as demonstrated in Remark~\ref{rem:discuss} and Proposition~\ref{gena} in \ref{appen2}, the condition is rarely satisfied in general cases.

\begin{figure}[ht]
 \centering
   \includegraphics[width=0.5\columnwidth]{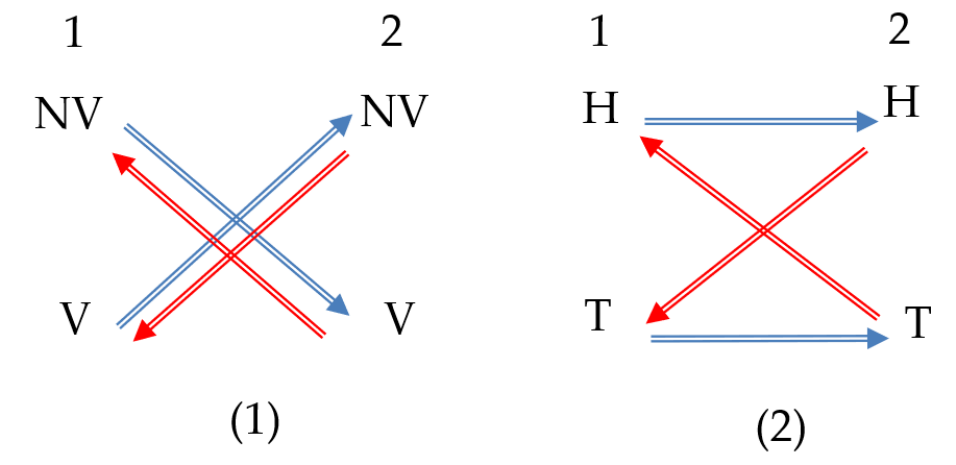}
   \caption{The graphs for two games}
   \label{fig:DIR}
\end{figure}

\begin{theorem}\label{Thmai2}
Consider a $2 \times 2$ game and $a_{i} \in A_{i}$ for some $i \in I$. If one of the following conditions is satisfied, then $q_{i,a_{i}}^{*} = \underline{\pi}_{i,a_{i}}$:
\begin{itemize}
\item[\textbf{(C1)}] $a_{i}$ is eventually non-serial, or

\item[\textbf{(C2)}] $|\tilde{\phi}_{i}(a_{i})|=1$ and $\tilde{\phi}_{j}(\tilde{\phi}_{i}(a_{i})) = \{a_{i}\}$.
\end{itemize}

Further, if $|Q(G,\Theta, p)|>1$ and both C1 and C2 are violated for one action, $q_{i,a_{i}}^{*}<\underline{\pi}_{i,a_{i}}$ for each $i \in I$ and $a_{i} \in A_{i}$.
\end{theorem}
The proof can be found in \ref{sec:proof}, which formalizes and generalizes the intuition in Example~\ref{exa:heu}.

The condition in Theorem~\ref{Thmai2} has two parts: C1 means the payoff matrix of player $i$ has order 1, that is, the game is non-generic. C2 is less restrictive, and is satisfied by many games intensively studied in the literature: for example, the chicken game (both symmetric as in Rapoport and Chammah \cite{rc66} and asymmetric as in Goeree et al. \cite{ghp16}, pp.~25--26), coordination games (Goeree et al. \cite{ghp16}, pp.~29--30, Zhang \cite{zb13}, Zhang and Hofbauer \cite{zh16}), and many (but not all) dominance-solvable games. However, a Matching-Pennies-style game (Ochs \cite{oc95}) never satisfies C2.\footnote{One may see some similarity between our condition and potential games (Monderer and Shapley \cite{ms96}). However, the two concepts are mathematically independent.}

We emphasize that the condition in Theorem~\ref{Thmai2} is only a sufficient condition when $|Q(G,\Theta, p)|=1$, which means that, in this case, if it is violated, both $q_{i,a_{i}}^{*}=\underline{\pi}_{i,a_{i}}$ and $q_{i,a_{i}}^{*}<\underline{\pi}_{i,a_{i}}$ are possible; indeed, the former happens in Example~\ref{ex11}, and the latter in Example~\ref{ex33}. Only when $|Q(G,\Theta, p)|>1$ the condition is both sufficient and necessary.

A direct implication of Theorem~\ref{Thmai2} is that when $|Q(G,\Theta, p)|=1$ and either C1 or C2 in Theorem~\ref{Thmai2} is satisfied, $C(G,\Theta,p) = \times_{j \in I} \Sigma_{j, \Delta^p}^\infty$. In that case, QRE provides a prediction that has a solid epistemic foundation: it could be reached only by individual strategic reasoning in each population.\footnote{The condition for the uniqueness of QRE has not been  fully studied yet. See Melo \cite{me21} for a survey.}

Note that the condition is only about $G$ and is independent of $\langle\Theta, p\rangle$. Consequently, changing the distribution or the value of the parameter does not alter the result. This characteristic may prove advantageous for experimental and empirical research in which only a fixed set of distributions can be practically implemented (in agents' beliefs). We leave the investigation of conditions on $p$ to the next section.

\begin{remark}\label{rem:discuss}
As mentioned above, the condition for the coincidence could be given for general two-person games. We leave the technical details to \ref{appen2} and only sketch the idea here. For each $a_{i}$, we denote by $R(a_{i})$ the set of actions that are directly or indirectly influenced by or influence it. The statement in Theorem~\ref{Thmai2} still holds if we replace (C2) by (C2'): $R(a_{i}) \neq A_{1} \cup A_{2}$. We can show that (C2) and (C2') are equivalent (Lemma \ref{strin}). Yet, as shown in Proposition~\ref{gena}, for a generic two-person game in which at least one player has more than two  actions, (C2) is never satisfied for any action. Therefore, the condition in Theorem~\ref{Thmai2} is quite stringent in generic games.
\end{remark}

\section{Rank-dependent choice equilibrium and \texorpdfstring{$\Delta^{M}$}{delta^M}-rationalizability}\label{rdceep}

\noindent So far, we have shown that transparency of the distribution profile $p$ of idiosyncratic payoff shocks in populations is the critical epistemic condition for QRE. Moreover, we presented a condition for the payoff structure that ensures that QRE coincides with the behavioral consequence characterizing the epistemic conditions. However, there may be concerns regarding the mixed nature of our result. On one hand, the independence from $p$ allows for broad applicability across diverse probability distributions of idiosyncratic payoff shocks. On the other hand,  applications in the literature predominantly rely on specific functional forms for $p$ or, more frequently, assume that $p$ is i.i.d. Consequently, our results may initially appear excessively general for direct applicability, and further research is warranted to align with these specific assumptions.

In this section, we attempt to solve this problem by investigating the epistemic condition behind the solution concept called \emph{rank-dependent choice equilibrium} first developed in Goeree et al. \cite{getal19} (see also Goeree and Louis \cite{gl21}). The emphasis pertains to a crucial attribute of the structural quantal response function: ordinal monotonicity.

\subsection{Rank-dependent choice equilibrium and the epistemic condition}\label{sec:cm}
\noindent Fix a static game $G = \langle I, (A_{j},u_{j})_{j \in I}\rangle$. For each $i \in I$ and for an arbitrary vector $v_{i} \in \mathbb{R}^{A_{i}}$, a \emph{ranking} of $v_{i}$, denoted by $r(v_{i})$, is a complete preorder $\leq_{i}$ on $A_{i}$ such that for each $a_{i}, a_{i}^{\prime} \in A_{i}$, $a_{i} \leq_{i} a_{i}^{\prime}$ if and only if $v_{i}(a_{i}) \leq v_{i}(a_{i}^{\prime})$.\footnote{Recall that a preorder on a set is a binary relation satisfying reflexivity and transitivity. A  preorder is complete iff every pair of elements in the set is comparable.} For each $\pi_{-i} \in \Delta(A_{-i})$, define $u_{i}(\pi_{-i})=(u_{i}(a_{i}, \pi_{-i}))_{a_{i} \in A_{i}}$, that is, the vector in $\mathbb{R}^{A_{i}}$ that  assigns to each action the expected payoff generated by it against $\pi_{-i}$. One can see that $r(u_{i}(\pi_{-i}))$ is a well-defined preorder on $A_{i}$. For $i \in I$, its \emph{rank-dependent response set} is defined as
\begin{equation*}
RDR_{i}(G) = \{\pi \in \times_{j \in I} \Delta(A_{j}): r(\pi_{i}) = r(u_{i}(\pi_{-i})) \}
\end{equation*}
Here, the condition $r(\pi_{i}) = r(u_{i}(\pi_{-i}))$ means that, given $\pi_{-i}$, the more payoff an action in $A_{i}$ generates, the more frequently it is used in $\pi_i$. Let $RDR(G) = \cap_{i \in I}RDR_{i}(G)$. Goeree et al. \cite{getal19} give the following definition.

\begin{definition}\label{def:rdce}
Given a static game $G$, the \textbf{rank-dependent choice equilibrium (RDCE)} is defined as
\begin{equation*}
RDCE(G) = cl(RDR(G))
\end{equation*}

\end{definition}

Goeree et al.'s \cite{getal19} Theorem~1 implies that $RDCE(G)$ is nonempty.\footnote{Our definition of ranking  is different from but equivalent to that in Goeree et al. \cite{getal19}. Goeree et al. \cite{getal19} focus on each action $a_{i} \in A_{i}$ and define $\rank(v_{i}(a_{i}))$ as the set of possible ranks of $v_{i}(a_{i})$ when the elements of $v_{i}$ are listed in non-increasing order. For example, for $v_{i}(a,b,c,d) = (4,2,1,2)$, $\rank(v_{i}(b)) =\rank(v_{i}(d))= \{2,3\}$. One can see that given $v_{i} \in \mathbb{R}^{A_{i}}$, $r(v_{i})$ can be directly generated from $(\rank(v_{i}(a_{i}))_{a_{i} \in A_{i}}$ and vice versa. Hence the two definitions are equivalent. We adopt our  definition since complete preorder provides an intuitive and manageable macrolevel property of action distributions consistent with the epistemic conditions that we are going to define.}\medskip

RDCE looks more ``chaotic'' than QRE because  the stability embodied in the equilibrium is built only upon ordinal monotonicity. In this equilibrium, the frequency-based order of actions mirrors the order of their associated payoffs, and no further refinement is needed to optimize  (in the objective sense) the choice.  This characteristic prompts a reevaluation of the fundamental assumptions of EGT, namely rationality and common belief in rationality, and their applicability within this context.

In fact, we can still assume that for each $i$, the idiosyncratic payoff shock follows a distribution $p_{i}$ (on $\Theta_{i} = \mathbb{R}^{A_{i}}$), but $p_{i}$ is no longer transparent. Each agent still has her conjecture $\mu^{i}$ about others' types and actions, and she still maximizes her payoff based on  her belief for the realized type. Therefore, we can still assume rationality and common belief in rationality. However, since agents now cannot agree on the specific form of $p$, $\Delta^{p}$ no longer provides the proper restriction on their beliefs. Hence, the question is which epistemic condition $\Delta$ is appropriate to capture RDCE. One can see in Definition~\ref{def:rdce} that, instead of a specific form of $p$, the essential principle is a macrolevel property  that prescribes better actions to be used more frequently. In the literature on QRE, this property is called \emph{monotonicity}, which is essential in the definition of regular QRE and in many applications (see, for example, Goree, Holt, and Palfrey \cite{ghp05}).

We now formulate this intuition. Recall that, as defined in Section~\ref{sec:qredef}, for each $a_{i} \in A_{i}$, $E_{i,a_{i}}(\pi_{-i})$ is the set of $\theta_{i} \in \Theta_{i}$ for which $a_{i}$ is a best response to $\pi_{-i}$. We say $p_{i} \in \Delta(\Theta_{i})$ is \emph{monotonic} iff the following condition holds for each $\pi_{-i} \in \Delta(A_{-i})$ and each $a_{i},a_{i}^{\prime} \in A_{i}$:
\begin{equation}\label{mmm}
 \tag{\textbf{M}}
 p_{i}(E_{i,a_{i}}(\pi_{-i})) \geq p_{i}(E_{i,a_{i}^{\prime}}(\pi_{-i})) \text{ if and only if } u_{i}(a_{i},\pi_{-i}) \geq u_{i}(a_{i}^{\prime},\pi_{-i})
 \end{equation}
 Condition (\ref{mmm}) can be decomposed into two parts: \textbf{(M1)} If $u_{i}(a_{i},\pi_{-i}) > u_{i}(a_{i}^{\prime},\pi_{-i})$, then $p_{i}(E_{i,a_{i}}(\pi_{-i})) > p_{i}(E_{i,a_{i}^{\prime}}(\pi_{-i}))$, and  \textbf{(M2)} If $u_{i}(a_{i},\pi_{-i}) = u_{i}(a_{i}^{\prime},\pi_{-i})$, then $p_{i}(E_{i,a_{i}}(\pi_{-i})) = p_{i}(E_{i,a_{i}^{\prime}}(\pi_{-i}))$. We use $M_{i}$ to denote the set of all monotonic distributions on $\Theta_{i}$ and let $M = (M_{j})_{j \in I}$. A belief $\mu^{i} \in \Delta(\Theta_{-i} \times A_{-i})$ is \emph{consistent with monotonicity} iff
 \begin{equation}\label{cm}
 \tag{\textbf{CM}}
 \mu^{i} = \prod_{j \neq i} \mu^{i}_{j} \text{ with } \mu_{j}^{i} \in \Delta(\Theta_{j} \times A_{j}) \text{ and } \text{marg}_{\Theta_{j}}\mu^{i}_{j} \in M_{j} \text{ for each } j \neq i
 \end{equation}
 As with condition (\ref{equ1}) in Section~\ref{deltapro}, (\ref{cm}) has two parts: the independence assumption and the consistency of each $\mu^i_j$'s marginal distribution on $\Theta_j$ with monotonicity. We use $\Delta^{M}_{i}$ to denote the set of all  of $i$'s beliefs satisfying (\ref{cm}), and we let $\Delta^{M} = (\Delta^{M}_{j})_{j \in I}$.\footnote{Actually, we can assume the transparency of beliefs in i.i.d. instead of the more general monotonicity assumption. Because Proposition~\ref{suppo}, which provides a critical step for showing the main result in this section (Theorem~\ref{state22}), can be proved even if we assume i.i.d. (yet Proposition~\ref{suppo} does not hold in general if we assume i.i.d. \emph{with some special functional form}, for example, extreme value distributions). Here, we adopt the assumption of monotonicity since we want to relate our results to the monotonicity axiom in the definition of the regular quantal response function, which plays a central role in Goeree et al. \cite{getal19}.} Here, the superscript $M$ indicates the crucial role of monotonicity in the restriction.\footnote{\label{second} We formulate $\Delta^{M}$ as a restriction on first-order beliefs. It can alternatively be formulated as one on second-order beliefs. We opt for the former formulation because it addresses the property of $p$, which aligns with the QRE literature, for example, assuming $p$ to be i.i.d. Also, in the literature on EGT, most works center on first-order restrictions, although some---for example, Perea \cite{pa2011}, Friedenberg \cite{fa19}, and Battigalli and Catonini \cite{bc23}---investigate conditions on higher-order beliefs.}
 
 Consistency with monotonicity looks similar to a condition called ``respect of the opponents' preferences'' in the literature on  proper rationalizability, which studies the epistemic foundation of Myerson's \cite{mr78} proper equilibrium (Schuhmacher \cite{sf99}, Asheim \cite{ag01}, Perea \cite{pa2011}).\footnote{See section~6 of Perea \cite{pa2012} for a survey.} The distinction between the two conditions lies in the role of monotonicity. In the literature on proper rationalizability, monotonicity is used as an instrumental concept to construct a ``(regulated) trembling sequence'' and the focus is on observing the outcomes when the sequence approaches the limit; there, beliefs are formulated on lexicographic probabilities, which draw upon nonstandard analysis (Blume, Brandenburger, and Dekel \cite{bbd91}, \cite{bbd91b}). In contrast, our condition pertains exclusively to the macrolevel property that ``better''  actions are employed more frequently within the framework of standard probability theory. Our primary concern lies in the behavioral consequences (with respect to rankings) of this basic condition.
 \medskip

 As in Section~\ref{deltapro}, we assume rationality (R), common belief in rationality (CBR), and transparency of $\Delta^{M}$  (TCM). The $\Delta^{M}$-rationalization procedure can be defined by replacing $\Delta^{p}$ in Definition~\ref{def2} by $\Delta^{M}$. The question is whether, in the absence of precise information about the (latent) distribution $p$, with only  monotonicity as a commonly believed macrolevel property, it remains feasible to derive any macrolevel characteristics  of the $\Delta^{M}$-rationalizable outcomes comparable with RDCE. As a parallel to the $\Delta^{p}$-rationalizable action distribution in Definition~\ref{norac}, we define a $\Delta^{M}$-rationalizable ranking.
\begin{definition}
We call a complete preorder $\leq_{i}$ an \textbf{$n$-$\Delta^{M}$-rationalizable ranking} iff there is $\mu_{i} \in \Delta(\Theta_{i} \times A_{i})$ with $\mu_{i}(\Sigma_{i,\Delta^{M}}^{n})=1$ and marg$_{\Theta_{i}}\mu_{i} \in M_{i}$ such that $\leq_{i}$ is a ranking of marg$_{A_{i}} \mu_{i}$. The ranking $\leq_{i}$ is called \textbf{$\Delta^{M}$-rationalizable} iff it is $n$-$\Delta^{M}$-rationalizable for each $n \in \mathbb{N}$.
\end{definition}

Before comparing RDCE and $\Delta^{M}$-rationalizable rankings, we use two examples to show heuristically how $\Delta^{M}$-rationalization procedure proceeds.

\subsection{Examples}\label{sec:iete}

\begin{example}\label{exame:count}

\begin{table}[ht]
\centering
 \begin{tabular}{c c c c}

 \hline
$1\setminus 2$ & $d$ & $e$ & $f$ \\ 
 \hline
$a$ & 0,0 & 0,1  & 3,1  \\ 

 $b$ & 1,1 & 1,0 & 1,1  \\ 
 
 $c$ & 2,1 & 2,1 & 2,0  \\
 \hline
\end{tabular}
\caption{A $3 \times 3$ game}
\label{game1}
\end{table}

Consider the game $G$ in Table~\ref{game1} (Perea \cite{pa2011}). One can see that $RDR(G) = \{\pi_{1} \in \Delta(A_{1}): \pi_{1}(c) > \pi_{1}(b) > \pi_{1}(a)\}  
\times \{\pi_{2} \in \Delta(A_{2}): \pi_{2}(d) > \pi_{2}(e) > \pi_{2}(f) \}$. Following the graphical representation introduced in Goeree et al. \cite{getal19}, the RDCE of this game is depicted by the orange area in Figure~\ref{fig:H1rt}.

\begin{figure}[ht] 
\centering
  \includegraphics[width=0.9\columnwidth]{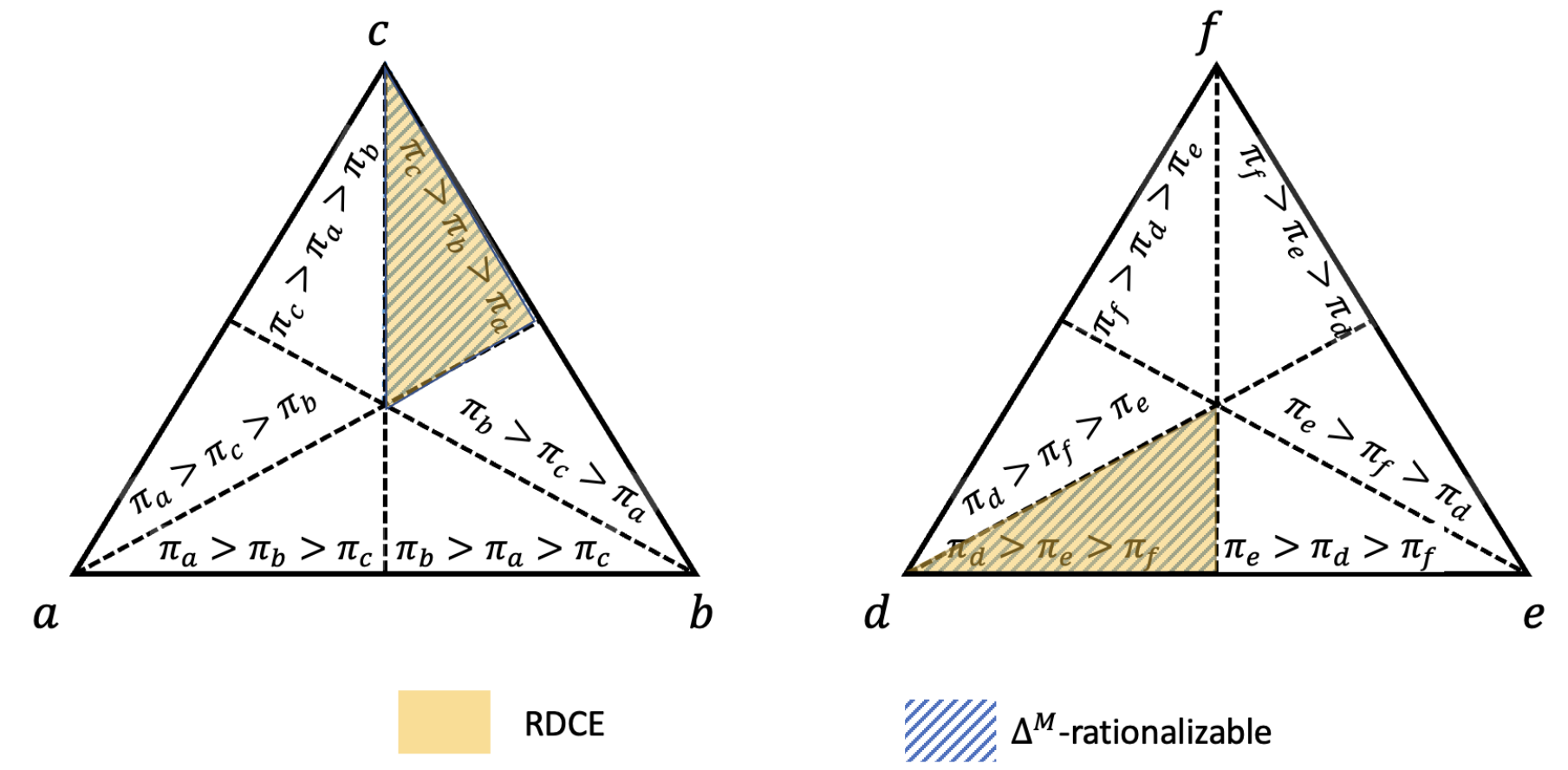}
  \caption{The R-rationalizability and RDCE equilibrium of $G$ in Example~\ref{exame:count}}
 \label{fig:H1rt}
\end{figure}

We now elucidate how the $\Delta^{M}$-rationalization procedure proceeds. At the beginning, players can believe anything. At \textbf{step 1}, since, for player 1, $b$ is dominated by $c$ (that is, $c$ generates a higher payoff than $b$ does for player 1 regardless of player 1's belief), any ranking incompatible with $c>_{1}b$ is not $1$-$\Delta^{M}$-rationalizable. At \textbf{step 2}, since (an agent from) player (population) 2 believes in rationality and her belief is in $\Delta^{M}$ for each of her beliefs $\mu^{2}$, marg$_{A_{1}}\mu^{2}(c) >$ marg$_{A_{1}}\mu^{2}(b)$, and, consequently, for player 2, $e$ generates a higher payoff than $f$, which means that any ranking incompatible with $e >_{2} f$ is not $2$-$\Delta^{M}$-rationalizable.

Drawing upon this, at \textbf{step 3}, for each belief $\mu^{1}$ of player 1, marg$_{A_{2}}\mu^{1}(e) >$ marg$_{A_{2}}\mu^{1}(f)$. Such a condition eliminates the possibility that the payoff generated from $a$ is as good as or better than that from $c$; indeed, if $u_{1}(a, $ marg$_{A_{2}}\mu^{1}) \geq u_{1}(c,$ marg$_{A_{2}}\mu^{1})$ holds, then marg$_{A_{2}}\mu^{1}(f) > \frac{2}{3}$, which is impossible under the condition marg$_{A_{2}}\mu^{1}(e) >$ marg$_{A_{2}}\mu^{1}(f)$. Therefore, any ranking compatible with $a\geq_{1}c$ is not $3$-$\Delta^{M}$-rationalizable.

Taking this into account, at \textbf{step 4}, player 2 cannot have any belief $\mu^{2}$ satisfying marg$_{A_{1}}\mu^{2}(a) \geq$ marg$_{A_{1}}\mu^{2}(c)$. More specifically, among the remaining possibilities, she has to eliminate every belief $\mu^{2}$ satisfying marg$_{A_{1}}\mu^{2}(a) \geq$ marg$_{A_{1}}\mu^{2}(c)>$ marg$_{A_{1}}\mu^{2}(b)$ (because all other possibilities compatible with marg$_{A_{1}}\mu^{2}(a) \geq$ marg$_{A_{1}}\mu^{2}(c)$ have been eliminated before). This makes $e >_{2} f \geq_{2} d$ not $4$-$\Delta^{M}$-rationalizable. 

Therefore, $c >_{1} a \geq_{1} b$ is not $5$-$\Delta^{M}$-rationalizable at \textbf{step 5} and  $e \geq_{2} d  >_{2} f$ is not $6$-$\Delta^{M}$-rationalizable at \textbf{step 6}. The procedure ends here. The only $\Delta^{M}$-rationalizable ranking is $c >_{1} b >_{2} a$ for 1 and $d >_{2} e >_{2} f$ for 2. The process is illustrated in Figure~\ref{fig:Elim1}, in which each ranking is located in the area in $\Delta(A_{i})$ corresponding to it; the number within each circle represents the step at which the ranking is eliminated.\footnote{Since the domain is clear, to simplify symbols, we omit the subscripts of the rankings in the figure.}  The $\Delta^{M}$-rationalizable rankings correspond to the dashed areas in Figure~\ref{fig:H1rt}. One can see that $RDR(G)$ is contained in (actually coincides with) the areas of the $\Delta^{M}$-rationalizable rankings. $\blacktriangle$

\begin{figure}[ht] 
\centering
  \includegraphics[width=0.9\columnwidth]{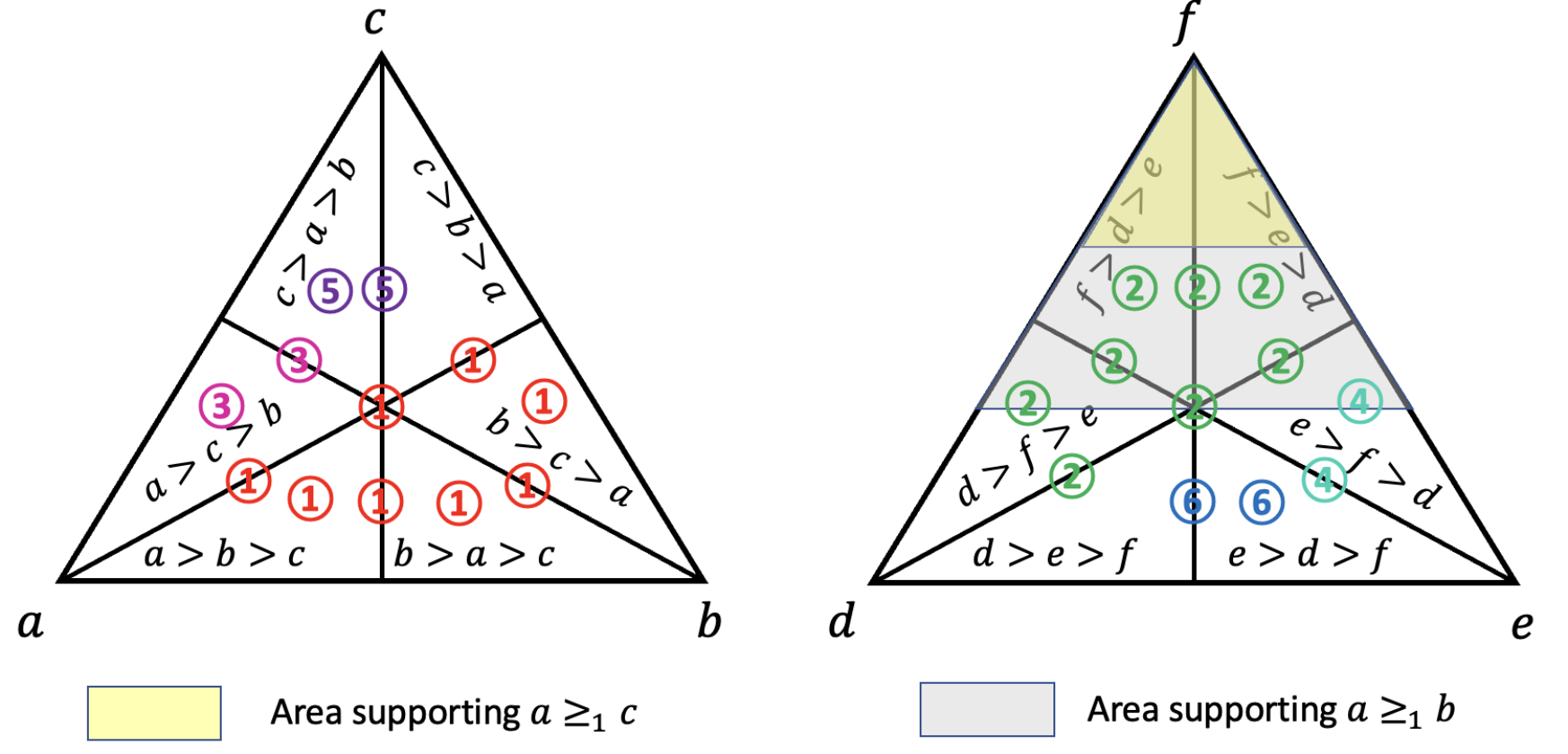}
  \caption{The elimination process to obtain $\Delta^{M}$-rationalizable rankings in Example~\ref{exame:count}}
 \label{fig:Elim1}
\end{figure}

\end{example}

\begin{example}\label{exam:liber}

\begin{table}[ht]
\centering
 \begin{tabular}{c c c c}

 \hline
$1\setminus 2$ & $b_{1}$ & $b_{2}$ & $b_{3}$ \\ 
 \hline
$a_{1}$ & $0, 0$ &  $15, -15$  & $-2,2$ \\ 

 $a_{2}$ & $-15, 15$ &  $0,0$ & $-1,1$ \\ 
 
 $a_{3}$ & $2,-2$ & $1,-1$ & $0,0$ \\
\hline
\end{tabular}
\caption{Liebermann's \cite{l60} $3 \times 3$ game}
\label{TABlibe}
\end{table}

Consider the game $G$ in Table~\ref{TABlibe} (Liebermann \cite{l60}). Goeree et al. \cite{getal19} calculated the $RDR(G)$ and $RDCE(G)$ of this game, and the latter is depicted by the orange area in  Figure~\ref{fig:Liberman} (1).\footnote{RDR here is essentially the interior of RDCE. Figure~\ref{fig:Liberman} only shows the results for player 1.  Since the game is symmetric, the results for player 2 can be obtained by replacing $a_{k}$ by $b_{k}$, $k = 1,2,3$. As shown in Goeree et al. \cite{getal19}, the curve in Figure~\ref{fig:Liberman} (1) running from the center to $a_{3}$ (the unique Nash equilibrium) depicts the logit equilibria with respect to the parameter running from $0$ to $\infty$.} By applying the argument in Example~\ref{exame:count}, one can see that the $\Delta^{M}$-rationalization process terminates at step 1: since $a_{2}$ ($b_{2}$) is dominated by $a_{3}$ ($b_{3}$), all rankings compatible with $a_{2} \geq a_{3}$ ($b_{2} \geq b_{3}$) are eliminated, and all other rankings are $\Delta^{M}$-rationalizable. Still, one can see that $RDR(G)$ is contained in the areas of the $\Delta^{M}$-rationalizable rankings. $\blacktriangle$

\begin{figure}[ht] 
\centering
  \includegraphics[width=0.9\columnwidth]{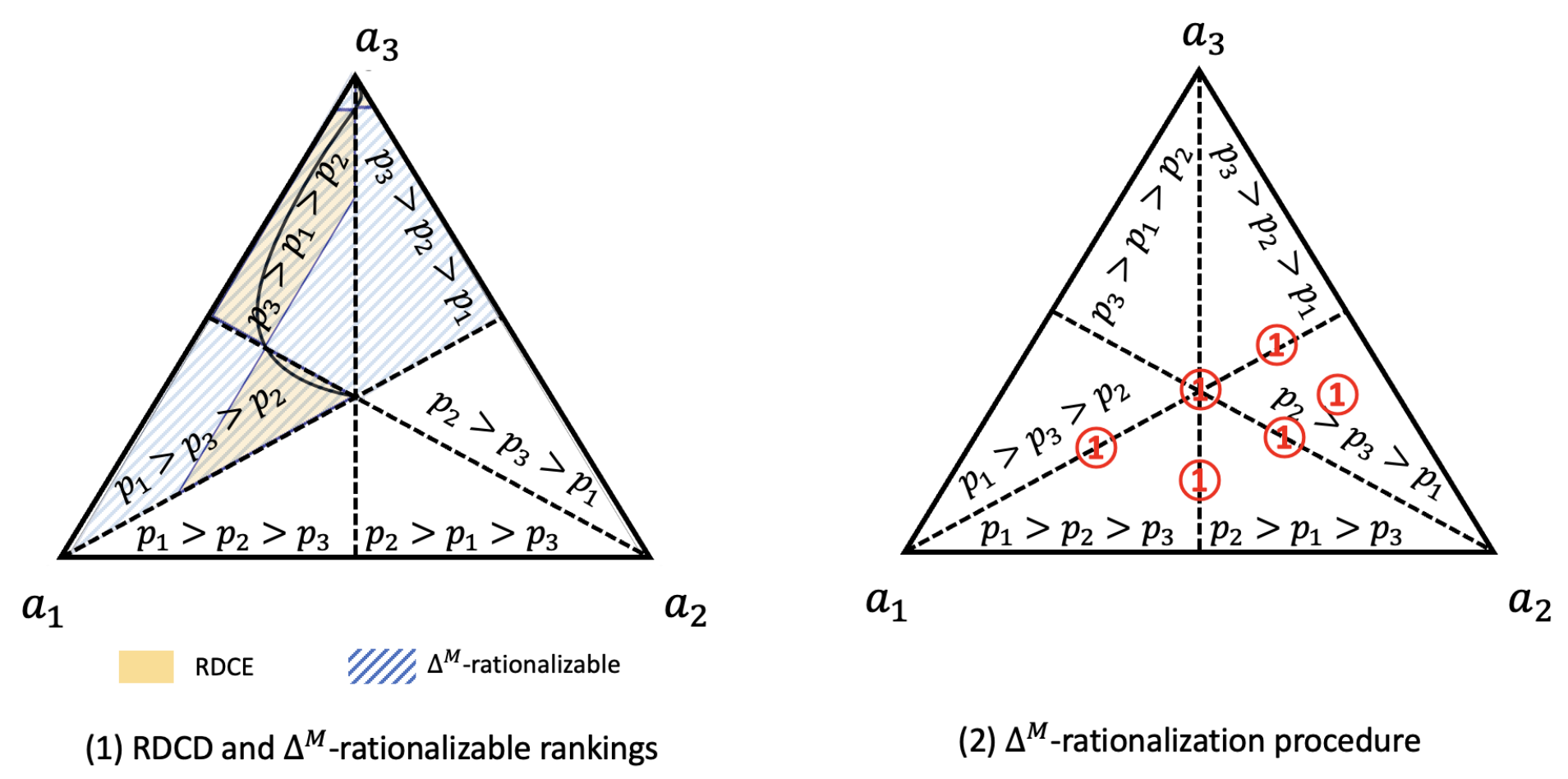}
  \caption{RDCE, $\Delta^{M}$-rationalizable rankings, and the elimination in Example~\ref{exam:liber}}
 \label{fig:Liberman}
\end{figure}
\end{example}

\subsection{Iterated elimination of rankings}\label{sec:ier}
\noindent In this subsection, for simplicity, we focus on two-person games; the definitions and results can be directly generalized. Fix a game $G = \langle A_{1}, A_{2}; u_{1}, u_{2}\rangle$. For each $i \in \{1,2\}$, we use $\Lambda_{i}$ to denote the set of complete preorders on $A_{i}$. For each $\lambda_{i} \in \Lambda_{i}$, we define $\Area(\lambda_{i}) = \{\pi_{i}\in \Delta(A_{i}) :\lambda_{i} \text{ is the ranking of }\pi_{i}\}$; for a subset $E_{i} \subseteq \Lambda_{i}$, we let $\Area(E_{i}) = \cup_{\lambda_{i} \in E_{i}}\Area(\lambda_{i})$. One can see that the domain of $\Area(\cdot)$ can be straightforwardly extended to include all preorders on each $A_{i}$. For example, given $A_{1} = \{a,b,c\}$ and the preorder $a\geq_{1} b$, $\Area(a\geq_{1} b) := \{\pi_{1} \in A_{1}: \pi_{1}(a) \geq \pi_{1}(b)\}$. A ranking $\lambda_{i}$ is \emph{compatible with} a preorder $\rho_{i}$ iff $\Area(\lambda_{i}) \cap \Area(\rho_{i}) \neq \emptyset$. Also, we define $\Cond(\lambda_{i}) = \{\pi_{j} \in \Delta(A_{j}): \lambda_{i} \text{ is the ranking of } u_{i}(\pi_{j})\}$; in words, $\Cond(\lambda_{i})$ is the set of the opponent $j$'s action distributions for which the ranking of the vector of payoffs generated by $i$'s actions has ranking $\lambda_{i}$. For example, for $\lambda_{i}:=a >_{i} b >_{i} c$, $\Area(\lambda_{i}) = \{\pi_{i} \in \Delta(\{a,b,c\}): \pi_{i}(a) > \pi_{i}(b) > \pi_{i}(c)\}$ and $\Cond(\lambda_{i}) = \{\pi_{j} \in  \Delta(A_{j}): u_{i}(a, \pi_{j}) > u_{i}(b, \pi_{j}) >u_{i}(c, \pi_{j}) \}$.

\begin{definition}
The procedure of \textbf{iterated elimination of rankings (IER)} is defined as follows:

\begin{itemize}

\item \textbf{Step 0}. $\Psi_{i, \Delta^{M}}^{0} = \Lambda_{i}$ for each $i \in I$.

\item \textbf{Step 1}. For each $i \in I$ and each $a_{i}, a_{i}^{\prime} \in A_{i}$, if $a_{i}$ dominates $a_{i}^{\prime}$, eliminate all rankings in $\Psi_{i}^{0}$ in which $a_{i} \leq_{i} a_{i}^{\prime}$; if $a_{i}$ weakly dominates $a_{i}^{\prime}$, eliminate all rankings in $\Psi_{i}^{0}$ in which  $a_{i} <_{i} a_{i}^{\prime}$. By doing this, we obtain $\Psi_{i, \Delta^{M}}^{1}$ for each $i \in I$.

\item \textbf{Step $n+1$}. Suppose that $\Psi_{j, \Delta^{M}}^{n}$ is defined for each $j \in I$. For each $\lambda_{i} \in \Psi_{i, \Delta^{M}}^{n}$, eliminate it if $\Cond(\lambda_{i}) \cap  \Area(\Psi_{j,\Delta^{M}}^{n}) = \emptyset$. By doing this, we obtain $\Psi_{i,\Delta^{M}}^{n+1}$ for each $i \in I$.

\end{itemize}
\end{definition}

We define $\Psi_{i,\Delta^{M}}^{*}(G)=\cap_{n=1}^{\infty} \Psi_{i,\Delta^{M}}^{n}$. It is straightforward to see that $\Psi_{i, \Delta^{M}}^{*}(G)$ is nonempty for each $i \in I$. 
Since each $\Lambda_{i}$ is finite and the procedure generates a non-increasing sequence, it actually stops within finite steps, that is, there is a finite $K$ such that $\Psi_{i,\Delta^{M}}^{K+1} = \Psi_{i,\Delta^{M}}^{K}$ for each $i \in I$. Therefore, the procedure can be described in a pseudocode in Algorithm~\ref{alg:ord}. 

\begin{algorithm}[ht!]
\caption{An algorithm for computing $(\Psi_{i,\Delta^{M}}^{*})_{i \in I, a_{i} \in A_{i}}$}\label{alg:ord}
\begin{algorithmic}
\State \textbf{start:} $\Psi_{i,\Delta^{M}}^{o} =\Delta_{i}^{1} = \Lambda_{i}$ for each $i \in I$ 
\State \textbf{remove } $\lambda_{i} \in \Delta_{i}^{1}$ \textbf{ if } $a_{i} \dom a_{i}^{\prime}$ \textbf{ and } $\lambda_{i}$ is compatible with $a_{i} \leq a_{i}^{\prime}$
\State \textbf{remove } $\lambda_{i} \in \Delta_{i}^{1}$ \textbf{ if } $a_{i} \wdom a_{i}^{\prime}$ \textbf{ and } $\lambda_{i}$ is compatible with $a_{i} <a_{i}^{\prime}$
\State Set $\Psi_{i,\Delta^{M}}^{1} = \Delta_{i}^{1}$
\State Set $n = 1$
\While{$\Psi_{\iota,\Delta^{M}}^{n} \neq \Psi_{\iota, \Delta^{M}}^{n-1}$ for some $\iota \in I$}
  \State Set $\Delta_{i}^{n+1} = \Psi_{i,\Delta^{M}}^{n}$
  \State \textbf{remove} $\lambda_{i} \in \Delta_{i}^{n+1}$ \textbf{ if } $\Cond(\lambda_{i}) \cap  \Area(\Psi_{j,\Delta^{M}}^{n}) = \emptyset$
  \State Set $\Psi_{i,\Delta^{M}}^{n+1}=\Delta_{i}^{n+1}$
  \State Set $n = n+1$
\EndWhile
\State Set $\Psi_{i,\Delta^{M}}^{*} = \Psi_{i,\Delta^{M}}^{n}$
\end{algorithmic}
\end{algorithm}

IER draws inspiration from Perea's \cite{pa2011} \emph{iterated addition of preference restrictions} (IAPR). However, the primary objective of IAPR is to explore proper-rationalizable \emph{actions} rather than rankings. In general, IAPR generates preorders, and the preorders serve an instrumental role to construct lexicographic beliefs and to find best responses to those lexicographic beliefs. Therefore, the application and interpretation of the rankings here differ from IAPR significantly.

The following statement shows that IER characterizes $\Delta^{M}$-rationalizability.

\begin{proposition}\label{suppo}
For each $i \in \{1,2\}$ and $n \in \mathbb{N}$, a ranking $\lambda_{i}$ is $n$-$\Delta^{M}$-rationalizable if and only if $\lambda_{i} \in \Psi_{i,\Delta^{M}}^{n}$. Consequently, $\lambda_{i}$ is $\Delta^{M}$-rationalizable if and only if $\lambda_{i} \in \Psi_{i,\Delta^{M}}^{*}$.
\end{proposition}

From Proposition~\ref{suppo}, we have the following result, which states that RDCE is essentially included in $\Delta^{M}$-rationalizable outcomes.

\begin{theorem}\label{state22}
$RDR(G) \subseteq  \Area(\Psi^{*}_{1, \Delta^{M}}(G)) \times \Area(\Psi^{*}_{2, \Delta^{M}}(G))$
\end{theorem}

\begin{corollary}
$RDCE(G) \subseteq cl(\Area(\Psi^{*}_{1, \Delta^{M}}(G)) \times \Area(\Psi^{*}_{2, \Delta^{M}}(G)))$
\end{corollary}

One may wonder whether the relation between RDCE and $\Delta^{M}$-rationalizability can be further refined. For example, from the two examples in Section~\ref{sec:iete}, one may conjecture that $\Delta^{M}$-rationalizable rankings are the contours  of RDR, that is, a ranking on $A_{i}$ is $\Delta^{M}$-rationalizable if and only if it is the ranking of some $\pi_{i} \in$ Proj$_{i}RDR(G)$. However, this statement does not hold in general. For example, consider the $2 \times 2$ Matching Pennies game studied in Goeree et al. \cite{getal19} (Table~13.4). For each player, every ranking is $\Delta^{M}$-rationalizable, while for the column player, $H > C$ does not correspond to any mixed action in Proj$_{\text{column}}RDR$. 

Compared to $\Delta^{p}$-rationalizability, the requirements imposed by $\Delta^{M}$-rationalizability are significantly relaxed. First, $\Delta^{M}$-rationalizability does not need common belief in $p$. Instead, the sole requisite is common belief in monotonicity, a condition that is reasonable and easier to satisfy. Second, since the set $\Lambda_{i}$ of rankings is finite, the procedure of IER terminates within finite steps, which avoids the  cognitively formidable challenge of computing an asymptotic sequence of minimums for each action as in the $\Delta^{p}$-rationalization procedure.

It is easy to find some trivial conditions for the (essential) coincidence of RDCE and $\Delta^{M}$-rationalizability. For example, the coincidence holds if for each player $i$ and each pair $a_{i},a_{i}^{\prime}$, either $a_{i}$ dominates $a_{i}^{\prime}$ or $a_{i}^{\prime}$ dominates $a_{i}$. Yet 
identifying a non-trivial condition as a counterpart to Theorem~2 appears to be a challenging task because $\Delta^{M}$-rationalizability is built upon a macrolevel  qualitative condition. We anticipate that  further research will  be conducted in this direction in the future.

Goeree et al. \cite{getal19} show that for finite games, RDCE essentially coincides with the set of all QREs  with different regular quantal response functions. One may wonder whether a similar relation holds between $\Delta^{p}$- and $\Delta^{M}$-rationalizability, that is, whether for each $i \in I$, the union of all $\Delta^{p}$-rationalizable action distributions (essentially) coincides with the areas of $\Delta^{M}$-rationalizable rankings for $i$. If we restrict each $p_{i}$ to monotonic distributions, then it is easy to see that for each $i$, each $\Delta^{p}$-rationalizable action distribution is within the area of some $\Delta^{M}$-rationalizable ranking of $i$. Here, monotonicity is leveraged at each step of the $\Delta^{p}$-rationalization procedure and plays a similar role as in the procedure of IER. However, the other way might not hold in general, because a best-response chain of action distributions in the areas of $\Delta^{M}$-rationalizable rankings may correspond to different distributions on $\Theta$ and those distributions might not be reconcilable.

\section{Concluding remarks}\label{concluding}

\noindent This paper investigated the epistemic foundations of QRE and its variant RDCE in static games with the methodology of $\Delta$-rationalizability. In addition to the canonical assumptions of rationality and common belief in rationality, by imposing  transparency of the distributions of idiosyncratic payoff shocks, we obtained a solution concept called $\Delta^{p}$-rationalizability, which includes action distributions generated from QRE. We also presented a condition of payoff structures under which QRE coincides with or provides tight information about the behavioral implications of the epistemic conditions. Then, to leverage the assumptions about $p$ in the literature and to relax the stringent epistemic condition of transparency of $\Delta^{p}$, we studied RDCE, a parameter-free variant of QRE, and showed that transparency of monotonicity of the distribution of the types provides an epistemic condition for RDCE. We also compared our two solution concepts and observed that, in contrast to the relationship between QRE and RDCE, the area of $\Delta^{p}$-rationalizable action distributions (for all monotonic $p$) forms a subset of the area of $\Delta^{M}$-rationalizable rankings, while the other way around does not necessarily hold.

Our concepts and results have the potential to offer valuable insights for interpreting data. Many studies adopt an outsider's perspective and interpret $p$ as the distribution of errors. By specifying the distribution to adhere to a particular functional form, such as the extreme value distribution, researchers employ statistical methodologies, such as maximum likelihood estimation, to demonstrate that experimental data can be elucidated as the behavior of errors approaching zero. Our results suggest that in some cases, we can regard $p$ as true and publicly believed distributions of some intrinsic attributes.\footnote{For instance, consider vaccination across diverse population groups. Statistical data pertaining to vulnerability, such as data on the percentage of individuals within specific groups having preexisting medical conditions, are frequently accessible and publicly available in various communities. Furthermore, individuals participating in experiments might hold certain beliefs that align with widely recognized statistical tendencies, for instance, the academic performance of students within a university. These beliefs can significantly affect their decision-making process, whether consciously or subconsciously.} The observed behavior could be interpreted as a procedure of iteratively eliminating unreasonable outcomes and reestimating each choice according to one's observation of others' behavior based on a transparent $p$.

Also, we may assume that some macrolevel property of $p$ instead of the accurate functional form of $p$ is commonly believed, which is realistic in many cases. In this vein, $\Delta^{M}$-rationalizability could be regarded as providing a reasonable benchmark for explaining data. For example, one can see that all outcomes in Liebermann \cite{l60} fall into the area of $\Delta^{M}$-rationalizable rankings. However,  the outcomes do not fall ``evenly'' into the three rankings' areas in Figure~\ref{fig:Liberman} (1): in Liebermann's \cite{l60} data, no behavior with ranking $p_{1} > p_{3} > p_{2}$ was recorded, even though the ranking is both $\Delta^{M}$-rationalizable and compatible with RDCE. The epistemic condition presented in this paper may necessitate further refinement, or  we may need to integrate our results with other established theories, such as dynamic learning, in order to provide a more comprehensive interpretation of the data.

\appendix
\addcontentsline{toc}{section}{Appendices}
\section*{Appendices}
\section{Proofs}\label{sec:proof}

\begin{proof}[Proof of Proposition~\ref{Thmai1}]
Fix $i \in I$. We have to construct a $\mu_{i} \in \Delta(\Theta_{i} \times A_{i})$ with $\mu_{i}(\Sigma_{i,\Delta^{p}}^{n})=1$ for each $n \in \mathbb{N}$ and marg$_{\Theta_{i}}\mu_{i} = p_{i}$ such that $\pi_{i} =$ marg$_{A_{i}}\mu_{i}$. 

For each $j \in I$, let $<_{j}$ be an asymmetric linear order on $A_{j}$; hence for each nonempty $A_{j}^{\prime} \subseteq A_{j}$, $\max A_{j}^{\prime}$ is well defined. Recall that for each $a_{j} \in A_{j}$, $E_{j,a_{j}}(\pi)$ is the set of $\theta_{j}$'s under which $a_{j}$ is a best response to $\pi_{-j}$. Now for each $\theta_{j} \in \Theta_{j}$, we define $A_{j, \theta_{j}}(\pi) =\{a_{j} \in A_{j}: \theta_{j} \in E_{j,a_{j}}(\pi)\}$, that is, the set of best responses of $j$ against $\pi_{-j}$ under $\theta_{j}$. It is clear that $A_{j, \theta_{j}}(\pi) \neq \emptyset$ for each $\theta_{j}$. Define a mapping $c_{j}:\Theta_{j} \rightarrow A_{j}$ such that $c_{j}(\theta_{j}) = \max A_{j, \theta_{j}}(\pi)$. Note that for a ``boundary'' $\theta_{j}$, that is, $|A_{j, \theta_{j}}(\pi)|>1$, $c_{j}(\theta_{j})$ is just a tie breaker; this does not cause any problem since $\pi_{j}$ is a probability distribution and those ``boundary'' payoff types form a null set with respect to $p_{j}$. 

It can be seen that $c_{j}$ is measurable, because for each $a_{j} \in A_{j}$, $c_{j}^{-1}(a_{j})$ is either $E_{j,a_{j}}$ or $E_{j,a_{j}} \setminus \cap_{a_{j}^{\prime} \in A_{j}^{\prime}}E_{j,a_{j}^{\prime}}$ for some $A_{j}^{\prime} \subseteq A_{j}$, both measurable. Let Gr$(c_{j}) = \{(\theta_{j},c_{j}(\theta_{j})): \theta_{j} \in \Theta_{j}\}$, that is, the graph of $c_{j}$. Now we define $\mu_{j} \in \Delta(\Theta_{j} \times A_{j})$. For each $X \subseteq \Theta_{j} \times A_{j}$, we say $X$ is measurable iff Proj$_{\Theta_{j}}[X \cap$ Gr$(c_{j})]$ is measurable in $\Theta_{j}$. For each measurable $X$, define $\mu_{j}(X) = p_{j}($Proj$_{\Theta_{j}}[X \cap$ Gr$(c_{j})])$. It is straightforward to see that marg$_{\Theta_{j}}\mu_{j} = p_{j}$ and for each $a_{j} \in A_{j}$, marg$_{A_{j}}\mu_{j}(a_{j}) = p_{j}[E_{j,a_{j}}(\pi)] = \pi_{j}(a_{j})$.

Let $\mu^{i} = \prod_{j \neq i}\mu_{j}$. Since $\pi$ is a QRE, it is clear that under each $\theta_{i}$, $c_{i}(\theta_{i})$ is a best response to $\mu^{i}$. This fixed-point property guarantees that $(\theta_{i},c_{i}(\theta_{i})) \in \Sigma_{i,\Delta^{p}}^{n}$ for each $n \in \mathbb{N}$. Therefore, $\pi_{i} =$ marg$_{A_{i}}\mu_{i}$ is $\Delta^{p}$-rationalizable.
\end{proof}

\begin{proof}[Proof of Proposition~\ref{prop2}]
It is straightforward to see that $q_{i,a_{i}}^{n}$ is a lower bound. Suppose that for some $q^{o}_{i,a_{i}} >q_{i,a_{i}}^{n}$, $p_{i}(a_{i}) \geq q_{i,a_{i}}^{o}$ holds for each $n$-$\Delta^{p}$-rationalizable distribution $p_{i}$. It implies that some $(\theta_{i},a_{i})$ not included in $E_{i,a_{i}}^{n}$ (see Algorithm~\ref{alg:cap}) is also $n$-$\Delta^{p}$-rationalizable, that is, at least one inequality in the definition of $E_{i,a_{i}}^{n}$ is violated. Yet this means that there is a belief allowable at step $n-1$ which does not support $a_{i}$ to be a best response under $\theta_{i}$. In this manner, we can find some $n$-$\Delta^{p}$-rationalizable distribution $p_{i}$ with $p_{i}(a_{i}) < q_{i,a_{i}}^{o}$, a contradiction.
\end{proof}

\begin{proof}[Proof of Proposition~\ref{PropDec}]
Fix $i \in I$ and $a_{i} \in A_{i}$. First, it is clear that for each $j \in I$ and $a_{j} \in A_{j}$, $q_{j,a_{j}}^{0} = 0$. Consider step $n \geq 0$ in the $\Delta^{p}$-rationalization procedure. Suppose that each $(q_{j,a_{j}}^{k})_{k=0}^{n}$ is nondecreasing. At step $n+1$, $q_{i,a_{i}}^{n+1}$ is determined by the set of $\theta_{i} \in \Theta_{i}$ for which only $(\theta_{i},a_{i})$ is $\Delta^{p}$-rationalizable, that is, the set of $\theta_{i}$ such that for each $a_{i}^{\prime} \in A_{i}$ and each $n$-th $\Delta^{p}$-rationalizable $q_{-i}$,  
\begin{equation}\label{optmax}
\theta_{i,a_{i}} - \theta_{i,a_{i}^{\prime}}  \geq  \sum_{a_{-i} \in A_{-i}} q_{-i}(a_{-i})[u_{i}(a_{i}^{\prime},a_{-i}) -u_{i}(a_{i},a_{-i})]
\end{equation}
The feasible region is determined by $q_{j,a_{j}}^{n} \leq q_{j}(a_{j}) \leq 1$ ($a_{j} \in A_{j}$) and $\sum_{a_{j}^{\prime}}q_{j}(a_{j}^{\prime})=1$. By inductive assumption, $q_{j,a_{j}}^{n-1} \leq q_{j,a_{j}}^{n}$ for each $j \neq i$ and $a_{j} \in A_{j}$. Hence the feasible region shrinks, which implies that $q_{i,a_{i}}^{n} \leq q_{i,a_{i}}^{n+1}$. Since $(q_{i,a_{i}}^{n})_{n \in \mathbb{N}}$ is nondecreasing and bounded within $[0,1]$, it converges to some  $q_{i,a_{i}}^{*}$. Since Theorem~\ref{Thmai1} shows that $\underline{\pi}_{i,a_{i}}$ is $\Delta^{p}$-rationalizable, it follows that $q_{i,a_{i}}^{*} \leq \underline{\pi}_{i,a_{i}}$.
\end{proof}

\begin{proof}[Proof of Theorem~\ref{Thmai2}]
Before showing Theorem~\ref{Thmai2}, we need a lemma.

\begin{lemma}\label{parti}
Consider $i \in \{1,2\}$ and distinct $a_{i},a_{i}^{\prime} \in A_{i}$. Either  $\phi_{i}(a_{i},a_{i}^{\prime}) =\phi_{i}(a_{i}^{\prime},a_{i}) =\emptyset$ or $\phi_{i}(a_{i},a_{i}^{\prime}) \cup\phi_{i}(a_{i}^{\prime},a_{i}) =A_{-i} $.
\end{lemma}
\begin{proof}
If $\phi_{i}(a_{i},a_{i}^{\prime}) =\emptyset$, then $u_{i}(a_{i}^{\prime},a_{-i})-u_{i}(a_{i},a_{-i}) = \overline{H}_{i}(a_{i},a_{i}^{\prime})$ for each $a_{-i} \in A_{-i}$, which implies that $u_{i}(a_{i},a_{-i}) -u_{i}(a_{i}^{\prime},a_{-i})$ is constant on $A_{-i}$, that is, $\phi_{i}(a_{i}^{\prime},a_{i}) =\emptyset$. Suppose that $\phi_{i}(a_{i},a_{i}^{\prime}) \neq\emptyset$. Since $u_{i}(a_{i},a_{j})-u_{i}(a_{i}^{\prime},a_{j})=-(u_{i}(a_{i}^{\prime},a_{j})-u_{i}(a_{i},a_{j}))$, it follows that $(A_{j} \setminus \phi_{i}(a_{i},a_{i}^{\prime}) )\subseteq \phi_{i}(a_{i}^{\prime},a_{i})$, and consequently $\phi_{i}(a_{i},a_{i}^{\prime}) \cup\phi_{i}(a_{i}^{\prime},a_{i}) =A_{j}$.
\end{proof}

First, suppose that $a_{i}$ is non-serial, that is, $u_{i}(a_{i}^{\prime},a_{-i})-u_{i}(a_{i},a_{-i}) = \overline{H}_{i}(a_{i},a_{i}^{\prime})$ for each $a_{-i} \in A_{-i}$ ($a_{i}^{\prime} \neq a_{i}$). For each $n \geq 0$, define $E_{i,a_{i}}^{n+1}$ to be set set of all $\theta_{i} \in \Theta_{i}$ satisfying
\begin{equation}\label{iterE}
\theta_{i,a_{i}} - \theta_{i,a_{i}^{\prime}}  \geq \max_{q_{-i} \text{ is }n\text{-th } \Delta^{p}\text{-rationalizable}} \sum_{a_{-i} \in A_{-i}}q_{-i}(a_{-i})[u_{i}(a_{i}^{\prime},a_{-i}) - u_{i}(a_{i},a_{-i})]
\end{equation}
It follows that here, $E_{i,a_{i}}^{n+1} = E_{i,a_{i}}^{1}$ $=\{\theta_{i}: \theta_{i,a_{i}}-\theta_{i,a_{i}^{\prime}} \geq \overline{H}_{i}^{a_{i},a_{i}^{\prime}}\}$ for each $n \geq 0$, and consequently $q_{i,a_{i}}^{*} = p_{i}(E^{1}_{i,a_{i}})$. Further, since for each $q \in \Delta(A_{1}) \times \Delta(A_{2})$, $E_{i,a_{i}}(q_{j}) = E_{i,a_{i}}^{1}$, hence $q_{i,a_{i}}^{*} = p_{i}(E_{i,a_{i}}(q_{j}))$, by Proposition~\ref{PropDec}, $q_{i,a_{i}}^{*} =\underline{\pi}_{i,a_{i}}$. Similarly, if $a_{i}$ is eventually non-serial, since the probabilities of their marginal actions will be fixed from the first round, $E_{i,a_{i}}^{n} = E_{i,a_{i}}^{2}$ for each $n \geq 2$, and we still have $q_{i,a_{i}}^{*}=\underline{\pi}_{i,a_{i}}$.

Now suppose that $a_{i}$ satisfies C2 with $a_{j}$ being the unique element in $\tilde{\phi}_{i}(a_{i})$ ($j \neq i$). Let $A_{i}=\{a_{i}, a_{i}^{\prime}\}$ and $A_{j}=\{a_{j}, a_{j}^{\prime}\}$. Lemma \ref{parti} implies that $\tilde{\phi}_{i}(a_{i}^{\prime}) = a_{j}^{\prime}$ and $\tilde{\phi}_{j}(a_{j}^{\prime}) = a_{i}^{\prime}$. We show that the following $q$ is a QRE:
\begin{equation*}
q_{i}(a_{i}) = q_{i,a_{i}}^{*}, q_{i}(a_{i}^{\prime})= 1- q_{i,a_{i}}^{*}, q_{j}(a_{j}) = q_{j,a_{j}}^{*}, q_{i}(a_{i}^{\prime} )= 1- q_{j,a_{j}}^{*}
\end{equation*}
First, at the limit, 
\begin{equation}
\label{poi}
q_{i,a_{i}}^{*} = p_{i}(\{\theta_{i}\in \Theta_{i}: \theta_{i,a_{i}}-\theta_{i,a_{i}^{\prime}} \geq (1-q_{j,a_{j}}^{*})\overline{H}_{i}(a_{i},a_{i}^{\prime}) + q_{j,a_{j}}^{*}\underline{H}_{i}(a_{i},a_{i}^{\prime})\})
\end{equation}
Here, as defined in Example~\ref{exa:heu}, $\underline{H}_{i}(a_{i},a_{i}^{\prime})$ is the difference of payoffs generated from $a_{i}^{\prime}$ and $a_{i}$ which is smaller than $\overline{H}_{i}(a_{i},a_{i}^{\prime}) $. Hence, if $q^{*}_{i,a_{i}}$ is the probability that $i$ uses $a_{i}$ in some QRE, it requires that $j$ uses $a_{j}$ with $q_{j,a_{j}}^{*}$, since at the limit,
\begin{equation}
\label{oyi}
 q_{j,a_{j}}^{*} = p_{j}(\{\theta_{j} \in \Theta_{j}: \theta_{j,a_{j}}-\theta_{j,a_{j}^{\prime}} \leq q_{i,a_{i}}^{*}\overline{H}_{j}(a_{j},a_{j}^{\prime}) + (1-q_{i,a_{i}}^{*})\underline{H}_{j}(a_{j},a_{j}^{\prime})\}).
\end{equation}
Hence $q_{i,a_{i}}^{*}$ and $q_{j,a_{j}}^{*}$ are mutually fulfilling. In addition, when $q_{j}(a_{j}^{\prime})= 1-q_{j,a_{j}}^{*}$, each $\theta_{i} \in \Theta_{i} \setminus \{\theta_{i}\in \Theta_{i}: \theta_{i,a_{i}}-\theta_{i,a_{i}^{\prime}} \geq (1-q_{j,a_{j}}^{*})\overline{H}_{i}(a_{i},a_{i}^{\prime}) + q_{j,a_{j}}^{*}\underline{H}_{i}(a_{i},a_{i}^{\prime})\}$ $ = \{\theta_{i}\in \Theta_{i}: \theta_{i,a_{i}^{\prime}}-\theta_{i,a_{i}} > (1-q_{j,a_{j}}^{*})(u_{i}(a_{i},a_{j}^{\prime})-u_{i}(a_{i}^{\prime},a_{j}^{\prime}) )+ q_{j,a_{j}}^{*}(u_{i}(a_{i},a_{j})-u_{i}(a_{i}^{\prime},a_{j}))\}$ and the measure of the latter set is $1-q_{i,a_{i}}^{*}$, that is, $p_{i}(E_{i,a_{i}^{\prime}})$. The same calculation works for $a_{j}^{\prime}$. Hence, $q$ is a QRE.

Suppose that $|Q(G,\Theta, p)|>1$ and both C1 and C2 are violated for one action. It follows that for distinct $\pi, \pi^{\prime} \in Q(G,\Theta, p)$, each $j \in I$ and $a_{j} \in A_{j}$, $\pi_{j}(a_{j})\neq \pi_{j}^{\prime}(a_{j})$. Therefore, $\sum_{a_{j}^{\prime} \in A_{j}}\underline{\pi}_{j,a_{j}^{\prime}}<1$ for each $j \in I$. Suppose that $q_{i,a_{i}}^{*}=\underline{\pi}_{i,a_{i}}$ for some $i \in I$ and $a_{i} \in A_{i}$. Since now every pair in $A_{1} \cup A_{2}$ are reachable from the other, if $a_{i}$ is used in some QRE by probability $q_{i,a_{i}}^{*}$, it requires each action $a_{j}$ is used by probability $q_{j,a_{j}}^{*}$, which does not form a probability distribution. Hence $q_{i,a_{i}}^{*}<\underline{\pi}_{i,a_{i}}$ for each $i \in I$ and $a_{i} \in A_{i}$.
\end{proof}

\begin{proof}[Proof of Proposition~\ref{suppo}]
Fix $i \in I$ and $n \in\mathbb{N}$ with $n \geq 2$.\footnote{As we noted in footnote \ref{second}, though we formulated the epistemic condition as a first-order one, it can be formulated and actually works as a secod-order condition on beliefs since it is about others' behavior (based on their beliefs). Therefore, the core procedure started from the second step.} It is straightforward to see that each $\lambda_{i}$ eliminated in IER at step $n$ is not $n$-$\Delta^{M}$-rationalizable since there is no belief to support them. We have to show that every order $\lambda_{i} \in \Psi_{i, \Delta^{M}}^{n}$  is $n$-$\Delta^{M}$-rationalizable. First, since $\lambda_{i}$ is not eliminated at step $n$, $\Cond(\lambda_{i}) \cap \Area(\Psi_{j,\Delta^{M}}^{n-1}) \neq \emptyset$, and for each $\pi_{j} \in \Cond(\lambda_{i}) \cap  \Area(\Psi_{j,\Delta^{M}}^{n-1})$, $\lambda_{i}$ is the ranking of $u_{i}(\pi_{j})$. Now the problem is whether there is a monotonic $p_{j}$ which generates $\pi_{j}$. We discuss two cases.

Case 1, $\Area(\lambda_{j}) \subseteq \Cond(\lambda_{i}) \cap  \Area(\Psi_{j,\Delta^{M}}^{n-1})$ for some ranking $\lambda_{j} \in \Psi_{j,\Delta^{M}}^{n-1}$, that is, any probability distribution on $\Delta(A_{j})$ which has the ranking $\lambda_{j}$ could support $\lambda_{i}$. Since $\lambda_{j} \in \Psi_{j,\Delta^{M}}^{n-1}$, there is an action distribution $\pi_{i}$ of $i$ whose ranking survives until step $n-2$ and $\lambda_{j}$ is the ranking of $u_{j}(\pi_{i})$. Since every distribution with ranking $\lambda_{j}$ supports $\lambda_{i}$, any i.i.d. $p_{j}$ (and the belief on actions which is pushed forward from it) supports $\lambda_{i}$ because, given $j$'s belief is $\pi_{i}$, the measure of $(E_{j,a_{j}}(\pi_{i}))_{a_{j} \in A_{j}}$ by $p_{j}$ has the ranking $\lambda_{j}$ and supports $\lambda_{i}$.

Case 2. The condition in Case 1 is not satisfied, that is, for each $\lambda_{j} \in \Psi_{j,\Delta^{M}}^{n-1}$, $\Area(\lambda_{j}) \nsubseteq \Cond(\lambda_{i}) \cap  \Area(\Psi_{j,\Delta^{M}}^{n-1})$. Let $\lambda_{j} \in \Psi_{j,\Delta^{M}}^{n-1}$ such that $\Area(\lambda_{j}) \cap  \Area(\Psi_{j,\Delta^{M}}^{n-1}) \neq \emptyset$. By Algorithm~\ref{alg:cap}, one can see that $\Area(\lambda_{j}) \cap  \Area(\Psi_{j,\Delta^{M}}^{n})$ is a polyhedron. Again, as in Case 1, since $\lambda_{j} \in \Psi_{j,\Delta^{M}}^{n-1}$, there is an action distribution $\pi_{i}$ of $i$ whose ranking survives until step $n-2$ and $\lambda_{j}$ is the ranking of $u_{j}(\pi_{i})$. Let $p_{i}$ be i.i.d. with each $p_{i, a_{i}}=q_{i}$; further, we can define $q_{i}$ by interpolation or by constructing an atom at inessential  subset(s) such that $p_{i} = q_{i}^{|A_{i}|}$ satisfies the restriction by the polyhedron;\footnote{For example, when the game is $2\times 2$, an extreme-value distribution works for $q_{i}$, because, for example, given $a > b$, no polyhedron which has nonempty interaction with $\Area(a>b)$ would prescribe $\pi_{b} > \frac{1}{2}$, and consequently an extreme-value distribution with parameter bigger than $0$ works. For larger games, extreme-value distributions are not enough, for example, in Example~\ref{exame:count}, step 3, $a\geq_{1} b$ cannot survive if we only allow extreme-value distributions, because, to support $a\geq_{1} b$ we need a distribution $\pi_{2} \in \Delta(A_{2}$ such that $\pi_{2}(f) > \frac{1}{3}$, while for $u_{2}(e) > u_{2}(f) > u_{2}(d)$ (the only ranking of $2$ surviving up to step 2 which is compatible with $\pi_{2}(f) > \frac{1}{3}$), every extreme-value distribution assigns a probability to $f$ equals to or smaller than $\frac{1}{3}$. In this case, we need some interpolation. For example, define $q_{2}$ with support on $[-6, 6]$ with an atom on $(3,6)$, or define $q_{2}$ by manipulating the uniform distribution on the interval ((or, more generally, let $q_{2}$ have support $[-\alpha K, \alpha K]$, where $K = \max_{j \in I, a,a^{\prime} \in A}|u_{j}(a) - u_{j}(a^{\prime})|$ and $\alpha >1$ and either have some atom on $(K, \alpha K)$ or manipulate the uniform distribution).} also, $p_{i} \in M_{i}$ since it is the product of i.i.d.
\end{proof}

\begin{proof}[Proof of Theorem~\ref{state22}]
The statement can be proved by induction. It is clear that for each $i \in I$ and $n \in \mathbb{N}$, a distribution in Proj$_{i}RDR(G)$ survives IER at step $n$ because the distribution in $\Delta_{j}(A_{j})$ survives at step $n-1$ (and vice versa). By Proposition~\ref{suppo}, it follows that Proj$_{i}RDR(G) \subseteq \Area(\Psi^{*}_{i, \Delta^{M}}(G))$ and consequently $RDR(G) \subseteq \Area(\Psi^{*}_{1, \Delta^{M}}(G)) \times \Area(\Psi^{*}_{2, \Delta^{M}}(G))$.
\end{proof}
\section{Generalization of Theorem~\ref{Thmai2}}\label{appen2}

We now formalize the idea in 2-person games. Consider a 2-person game $G=\langle A_{1},A_{2},u_{1},u_{2}\rangle$. Still, we call an action $a_{ik} \in A_{i}$ \emph{non-serial} iff $\phi(a_{ik})= \emptyset$; it is called \emph{eventually non-serial} iff for some $a_{j} \in \Phi(a_{i})$, $\Phi(a_{j} )= \emptyset$. We have the following result.

\begin{lemma}
If for some $i \in \{1,2\}, a_{i} \in A_{i}$ is non-serial, then each action in $A_{i}$ is non-serial, and consequently every action in the game is eventually non-serial.
\end{lemma}
\begin{proof}
If $a_{i}$ is non-serial, then for each $a_{i}^{\prime} \in A_{i}$ and each $a_{j} \in A_{j}$, $u_{i}(a_{i}^{\prime},a_{j})-u_{i}(a_{i},a_{j}) = \overline{H}_{i}(a_{i},a_{i}^{\prime})$, which implies that for each $a_{i}, a_{i}^{\prime} \in A_{i}$, $u_{i}(a_{i}^{\prime},a_{j})-u_{i}(a_{i},a_{j})$ is constant on $ A_{j}$. Hence $\phi_{i}(a_{i}) = \emptyset$ for each $a_{i} \in A_{i}$, and consequently every action in the game is eventually non-serial.
\end{proof}

We call a game \emph{serial} iff no action is non-serial.

For each $a_{i}, a_{j} \in A_{1} \cup A_{2}$, $a_{j}$ is \emph{\textbf{reachable from}} $a_{i}$, denoted by $a_{i} \rhd a_{j}$, iff there are $a_{i_{0},k_{0}},a_{i_{1},k_{1}},...,a_{i_{N},k_{N}} \in A_{1} \cup A_{2} (N \geq 0)$ such that $a_{i} = a_{i_{0},k_{0}}\Rightarrow a_{i_{1},k_{1}}\Rightarrow...\Rightarrow a_{i_{N},k_{N}} =a_{j}$. Note that $a_{i}$ is reachable from itself (that is, when $N=0$). Note that, first, $a_{i}$ is influenced by every action reachable from it, and second, for every action reachable from $a_{i}$, every action which can reach it also matters since, if we want to fix the probability of using $a_{i}$ as $q_{i,a_{i}}^{*}$, all those actions $b_{k}$ should also be used with $q_{k, b_{k}}^{*}$ and cannot be manipulated freely. We use $R(a_{i})$ to denote the set of all those actions, namely $R(a_{i}):=\{b_{k} \in A_{1} \cup A_{2}: a_{i} \rhd b_{k} \text{ or } b_{k} \rhd c_{s} \text{ for some }  c_{s} \text{ with }a_{i} \rhd c_{s}\}$.

\begin{proposition}\label{propsa}
Let $a_{i} \in A_{i}$. If one of the following condition is satisfied, then $q_{i,a_{i}}^{*} = \underline{\pi}_{i,a_{i}}$:

(\textbf{C1}) $a_{i}$ is eventually non-serial, or

(\textbf{C2'}) $R(a_{i}) \neq A_{1} \cup A_{2}$.
\end{proposition}

\begin{proof}
The eventually non-serial case is the same as in the proof of Theorem~\ref{Thmai2}. Now suppose that $a_{i}$ is not eventually non-serial and  $R(a_{i}) \neq A_{1} \cup A_{2}$. We define the following symbols: for each $j \in \{1,2\}$,
\begin{displaymath}
A_{j}^{o} := A_{j} \cap R(a_{i}), {\ \ } A_{j}^{\prime} := A_{j} \setminus R(a_{i})
\end{displaymath}
It is clear that for each $k \in I$, $A_{k}^{o}$ and $A_{k}^{\prime}$ form a partition of $A_{j}$. We have the following lemma.
\begin{lemma}\label{pius}
When $a_{ik}$ is not eventually non-serial, $R(a_{i}) \neq A_{1} \cup A_{2}$ implies that $A_{j}^{\prime} \neq \emptyset$ for each $j = 1,2$. 
\end{lemma}

To see this, suppose that $A_{j}^{\prime} = \emptyset$ for some $j \in \{1,2\}$. Then $A_{-j}^{\prime} \neq \emptyset$ since $A_{1}^{\prime} \cup A_{2}^{\prime} = (A_{1} \cup A_{2})\setminus R(a_{i})$. Yet since no $a_{-j} \in A_{-j}^{\prime}$ is non-serial, there should be some $b_{j} \in A_{j}$ marginal to some $a_{-j} \in A_{-j}^{\prime}$ which is not in $R(a_{i})$, otherwise by definition $A_{-j}^{\prime} = \emptyset$. Yet since $A_{j}^{\prime} = \emptyset$, $b_{j} \in R(a_{i})$, and consequently $a_{-j} \in R(a_{i})$, which implies that $A_{-j}^{\prime} = \emptyset$, a contradiction. Therefore, $A_{j}^{\prime} \neq \emptyset$ for both $j = 1,2$.

Now we return to the proof of Proposition~\ref{propsa}. To show $q_{i,a_{i}}^{*} = \underline{\pi}_{i,a_{i}}$, we show that Statement \ref{st1} holds here, that is, there is some QRE where $a_{i}$ is used by probability $q_{i,a_{i}}^{*}$.Consider $q=(q_{1},q_{2})\in \Delta(A_{1}) \times \Delta(A_{2})$ defined as follows:

(a) For each $a_{j} \in A_{j}^{o}$ for each $j \in \{1,2\}$, let $q_{j}(a_{j})=q_{j,a_{j}}^{*}$.

(b) Since we have shown above that $A_{j}^{\prime} \neq \emptyset$ for both $j = 1,2$, we can define 
\begin{displaymath}
B_{j} =\{r_{j}\in [0,1]^{A_{j}^{\prime}}: \sum_{a_{j} \in A_{j}^{\prime}} r_{j}(a_{j}) = 1-\sum_{t:a_{j} \in A_{j}^{o}}q_{j,a_{j}}^{*} \text{ and }r_{j}(a_{j}) \geq q_{j,a_{j}}^{*} \text{ for each }s \text{ with }a_{js} \in A_{j}^{\prime}\}
\end{displaymath}
Each $B_{j}$ is well defined. It is clear that each $B_{j}$ is compact and convex, and so is $B :=B_{1} \times B_{2}$. Now consider $g: B \rightarrow B$ such that for each $j \in \{1,2\}$, $r \in B$, and  $a_{j} \in A_{j}^{\prime}$, 
\begin{displaymath}
g_{j,a_{j}}(r)=p_{j}[E_{j}^{s}((q_{-j,a_{-j}}^{*})_{a_{-j} \in A_{-j}^{o}}, (r_{-j,a_{-j}^{\prime}})_{a_{-j}^{\prime} \in A_{-j}^{\prime}})],
\end{displaymath}
where $E_{j}^{s}$ is a measurable function from $[0,1]^{A_{-j}}$ to $2^{\Theta_{j}}$ such that for each $\gamma \in [0,1]^{A_{-j}}$, 
\begin{displaymath}
E_{j}^{a_{i}}(\gamma) =\cap_{a_{i}^{\prime} \neq a_{i}}\{\theta_{i}\in \Theta_{i}: \theta_{i,a_{i}} - \theta_{i,a_{i}^{\prime}}  \geq \sum_{a_{j} \in A_{j}} \gamma(a_{j})[u_{i}(a_{i}^{\prime},a_{j}) -u_{i}(a_{i},a_{j})]\}
\end{displaymath}
By our assumptions about $p_{j}$, $j = 1,2$, $g$ is continuous. It follows from Brouwer's  fixed point theorem that $g$ has a fixed point $r^{*}$. Then for each $a_{j} \in A_{j}^{\prime}$, we let $q_{j}(a_{j}) = r^{*}_{j,a_{j}}$. It can be seen that $q$ is a QRE in which $a_{i}$ is used by probability $q_{i,a_{i}}^{*}$.
\end{proof}

Now we discuss the contingency. The following result shows that (C2') implies (C2).

\begin{lemma}\label{strin}
Suppose that $a_{i} \in A_{i}$ is not eventually non-serial. Then if it satisfies condition (C2') in Proposition~\ref{propsa}, it satisfies (C2).
\end{lemma}

\begin{proof}
First, since $a_{i}$ is serial, $|\phi_{i}(a_{i})|\geq 1$. Suppose that $ |\phi_{i}(a_{i})|> 1$ and let $b_{j},b_{j}^{\prime} \in \phi_{i}(a_{i})$ with $b_{j}\neq b_{j}^{\prime}$. Then Lemma \ref{parti} implies that $\phi_{j}(b_{j}) \cup \phi_{j}(b_{j}^{\prime} )= A_{i}$, and consequently $A_{i} \subseteq R_{i}(a_{i})$. By Lemma \ref{pius}, it follows that $R(a_{i}) = A_{1} \cup A_{2}$, a contradiction. Hence, $|\phi_{i}(a_{i})|=1$.

Second, let $\phi_{i}(a_{i}) = \{b_{j}\}$. Suppose that there is some $a_{i}^{\prime} \in \phi_{j}(b_{j})$ with $a_{i}^{\prime} \neq a_{i}$. By Lemma \ref{parti}, it follows that $\phi_{i}(a_{i}) \cup \phi_{i}(a_{i}^{\prime}) = A_{j}$, and consequently $A_{j} \subseteq R(a_{i})$. By Lemma \ref{pius} it follows that $R(a_{i})= A_{1} \cup A_{2}$, a contradiction. Hence, $\phi_{j}(\phi_{i}(a_{i})) = \{a_{i}\}$.
\end{proof}

Lemma \ref{strin} shows that (C2') can be quite strict in general 2-person games. We have the following result, which, combining with Lemma \ref{strin}, implies that generic games never satisfy the condition. 

\begin{proposition}\label{gena}
Consider a 2-person serial game with $|A_{i}|\geq 2$ for each $i = 1,2$ and at least one player has more than two actions. Then no action satisfies (C2).
\end{proposition}

\begin{proof}
Without loss of generality, suppose that player 1 has more than two actions, and $a_{1}\in A_{1}$ satisfies (C2). By Lemma \ref{strin}, $\phi_{1}(a_{1}) = \{a_{2}\}$ for some $a_{2} \in A_{2}$. So $\{a_{1},a_{2}\} \subseteq R(a_{1})$. Let $b_{1},c_{1} \in A_{1}\setminus \{a_{1}\}$ be distinct. By Lemma \ref{parti}, $\phi_{1}(b_{1}) \cup \phi_{1}(c_{1}) = A_{2}$. So $a_{2} \in \phi_{1}(b_{1})$ or $a_{2} \in \phi_{1}(c_{1})$. Hence $\{b_{1},c_{1}\}\cap R(a_{1})\neq \emptyset$. By Lemma  \ref{strin}, $\{b_{1},c_{1}\} \subseteq \phi_{2}(b_{2})$ for each $b_{2} \in A_{2}$ with $b_{2} \neq a_{2}$ (since $|A_{2}| \geq 2$, such $b_{2}$ exists), and consequently $b_{2} \in R(a_{1})$. Hence $A_{2} \cap R(a_{1}) = A_{2}$. By Lemma \ref{pius}, $R(a_{1}) =  A_{1} \cup A_{2}$, a contradiction.
\end{proof}

\end{document}